\documentclass[letterpaper,11pt]{article}
\usepackage[utf8]{inputenc}

\pdfoutput=1

\usepackage{amsmath, amsthm, amssymb}
\usepackage[linesnumbered,vlined]{algorithm2e}
\usepackage[table,xcdraw]{xcolor}
\usepackage{mathtools}
\usepackage[numbers]{natbib}
\usepackage{comment} 
\usepackage{color-edits} 
\usepackage[most]{tcolorbox}
\usepackage{xfrac}
\usepackage{hyperref}
\usepackage{multirow}
\usepackage{caption}
\usepackage{bm}
\usepackage{newfloat}
\usepackage{enumitem}
\usepackage{xpatch}
\usepackage{soul}
\usepackage{bbm}
\usepackage{tikz}
\usepackage{titlesec}

\usepackage{fancyhdr}

\makeatletter
\newcommand{\thickhline}{%
    \noalign {\ifnum 0=`}\fi \hrule height 1.5pt
    \futurelet \reserved@a \@xhline
}
\makeatother

\usepackage[margin=1in]{geometry}

\allowdisplaybreaks

\definecolor{mygreen}{RGB}{20,120,60}

\title{Stochastic Vertex Cover with Few Queries}

\author{
Soheil Behnezhad\footnote{Part of the work was done while the first and the third authors were interns at TTIC.}\\
{\em Stanford University}
\and 
Avrim Blum\thanks{This work was supported by the National Science Foundation under grant CCF-1733556.} \\
{\em TTIC}
\and 
Mahsa Derakhshan\footnotemark[1]\\
{\em Princeton University}
}

\date{}

\newcommand{\E}[0]{\ensuremath{\mathbb{E}}}

\DeclareMathOperator{\Var}{Var}
\DeclareMathOperator{\Cov}{Cov}

\newcommand{\MM}[1]{\ensuremath{\mathsf{MM}(#1)}}

\newcommand{\cru}[1]{\ensuremath{G}}
\newcommand{\inm}[1]{\ensuremath{Z}}

\newcommand{\Stau}[0]{\ensuremath{\epsilon^2 p}}
\SetKwRepeat{Do}{do}{while}

\DeclareMathOperator*{\argmax}{arg\,max}
\renewcommand{\b}[1]{\ensuremath{\bm{\mathrm{#1}}}}
\DeclareMathOperator{\poly}{poly}
\DeclareMathOperator{\polylog}{polylog}

\DeclareMathOperator{\var}{Var}

\renewcommand{\epsilon}[0]{\ensuremath{\varepsilon}}

\let\originalleft\left
\let\originalright\right
\renewcommand{\left}{\mathopen{}\mathclose\bgroup\originalleft}
\renewcommand{\right}{\aftergroup\egroup\originalright}

\addauthor{sb}{blue}    
\addauthor{md}{red}    

\newtheorem{theorem}{Theorem}[section]
\newtheorem{result}{Result}
\newtheorem{question}{Question}

\newtheorem{lemma}{Lemma}[section]

\newtheorem{proposition}[lemma]{Proposition}
\newtheorem{corollary}[lemma]{Corollary}

\newtheorem{definition}[lemma]{Definition}
\newtheorem{claim}[lemma]{Claim}

\newtheorem{observation}[lemma]{Observation}

\makeatletter
\def\thm@space@setup{%
  \thm@preskip= 0.2cm
  \thm@postskip=\thm@preskip 
}
\makeatother

\definecolor{mygreen}{RGB}{20,155,20}
\definecolor{myred}{RGB}{195,20,20}
\definecolor{linkcolor}{RGB}{0,0,230}
\definecolor{mylightgray}{RGB}{230,230,230}
\definecolor{verylightgray}{RGB}{240,240,240}
\definecolor{commentcolor}{RGB}{120,120,120}

\SetCommentSty{mycommfont}

\newcommand{\smparagraph}[1]{
\par\addvspace{0.2cm}
\noindent \textbf{#1}
}

\hypersetup{
     colorlinks=true,
     citecolor= mygreen,
     linkcolor= myred,
     urlcolor= mygreen
}

\newcommand{\etal}[0]{\textit{et al.}}

\newcommand{\mc}[1]{\ensuremath{\mathcal{#1}}}

\newcounter{myalgctr}

\newenvironment{tbox}{
\par\addvspace{0.2cm}
\begin{tcolorbox}[width=\textwidth,
                  enhanced,
                  boxsep=2pt,
                  left=1pt,
                  right=1pt,
                  top=4pt,
                  boxrule=1pt,
                  arc=0pt,
                  colback=white,
                  colframe=black,
                  unbreakable
                  ]
}{
\end{tcolorbox}
}

\newenvironment{tboxh}{
\par\addvspace{0.2cm}
\begin{tcolorbox}[width=\textwidth,
                  enhanced,
                  boxsep=2pt,
                  left=1pt,
                  right=1pt,
                  top=4pt,
                  boxrule=1pt,
                  arc=0pt,
                  colback=white,
                  colframe=black,
                  unbreakable,
                  float=t
                  ]
}{
\end{tcolorbox}
}

\newenvironment{graytbox}{
\par\addvspace{0.1cm}
\begin{tcolorbox}[width=\textwidth,
                  enhanced,
                  frame hidden,
                  boxsep=5pt,
                  left=1pt,
                  right=1pt,
                  top=2pt,
                  bottom=2pt,
                  boxrule=1pt,
                  arc=0pt,
                  colback=mylightgray,
                  colframe=black,
                  breakable
                  ]
}{
\end{tcolorbox}
}

\newcommand{\tboxhrule}[0]{\vspace{0.1cm} \hrule \vspace{0.2cm}}

\newenvironment{titledtbox}[1]{\begin{tbox}#1 \tboxhrule}{\end{tbox}}
\newenvironment{titledtboxh}[1]{\begin{tboxh}#1 \tboxhrule}{\end{tboxh}}

\newenvironment{tboxalg2e}[1]{
\refstepcounter{myalgctr}
	\begin{titledtbox}{\textbf{Algorithm \themyalgctr.} #1}
	\vspace{-0.2cm}
}
{
	\vspace{-0.3cm}
	\end{titledtbox}
}

\newenvironment{highlighttechnical}[0]{
\vspace{0.1cm}
\begin{tcolorbox}[width=\textwidth,
                  enhanced,
                  boxsep=2pt,
                  left=1pt,
                  right=1pt,
                  top=4pt,
                  boxrule=0.8pt,
                  arc=0pt,
                  colback=mylightgray,
                  colframe=black,
                  unbreakable
                  ]
}{
\end{tcolorbox}
}

\begin{document}
\setlength{\parskip}{0.2em}

\maketitle

\begin{abstract}
\setlength{\parskip}{0.3em}

We study the minimum vertex cover problem in the following stochastic setting. Let $G$ be an arbitrary given graph, $p \in (0, 1]$ a parameter of the problem, and let $G_p$ be a random subgraph that includes each edge of $G$ independently with probability $p$. We are unaware of the realization $G_p$, but can learn if an edge $e$ exists in $G_p$ by querying it. The goal is to find an approximate minimum vertex cover (MVC) of $G_p$ by querying few edges of $G$ non-adaptively.

\medskip
This stochastic setting has been studied extensively for various problems such as minimum spanning trees, matroids, shortest paths, and matchings. To our knowledge, however, no non-trivial bound was known for MVC prior to our work. In this work, we present a:
\begin{itemize}
	\item $(2+\epsilon)$-approximation for \textbf{general} graphs which queries $O(\frac{1}{\epsilon^3 p})$ edges per vertex, and a
	\item $1.367$-approximation for \textbf{bipartite} graphs which queries $\poly(1/p)$ edges per vertex.
\end{itemize}

Additionally, we show that at the expense of a triple-exponential dependence on $p^{-1}$ in the number of queries, the approximation ratio can be improved down to $(1+\epsilon)$ for bipartite graphs.

\medskip
Our techniques also lead to improved bounds for bipartite stochastic matching. We obtain a $0.731$-approximation with nearly-linear in $1/p$ per-vertex queries. This is the first result to break the prevalent $(2/3 \sim 0.66)$-approximation barrier in the $\poly(1/p)$ query regime, improving algorithms of [Behnezhad~\etal{}, SODA'19] and [Assadi and Bernstein, SOSA'19].
\end{abstract}
\clearpage

{

\hypersetup{
     linkcolor= black
}

\tableofcontents{}
\clearpage
}

\clearpage

\section{Introduction}\label{sec:intro}

We study the following {\em stochastic vertex cover} problem. Let $G=(V, E)$ be a given $n$-vertex graph, $p \in (0, 1]$ a parameter of the problem, and let $G_p \subseteq G$ be a random subgraph that includes each edge in $E$ independently with probability $p$. We are unaware of the realization of $G_p$, but can learn if an edge $e \in E$ is realized in $G_p$ by {\em querying} it. The goal is to find an approximate minimum vertex cover (MVC) of $G_p$ by querying, non-adaptively, few edges in $G$.

This stochastic setting has been studied extensively over the last two decades for various problems such as minimum spanning trees and matroids \cite{DBLP:conf/soda/GoemansV04, DBLP:journals/rsa/GoemansV06}, packing problems \cite{YM18}, shortest paths \cite{DBLP:journals/rsa/Vondrak07}, and most relevant to our work, matchings \cite{blumetal,blumetalOR,AKL16,AKL17,BR18,YM18,sagt19,sosa19,soda19,stoc20,focs20}. There has also been quite a lot of related work on  ``network reliability'' in random subgraphs; see the book \cite{colbourn1987combinatorics} for some classic results of the 1980's as well as \cite{DBLP:conf/stoc/Karger20,DBLP:journals/siamcomp/GuoJ19} and the references therein for more recent works. While this is by no means a comprehensive list of all the related works, we are, to our knowledge, the first to consider a {\em covering} problem in the setting.

It would be useful to overview the known bounds for the matching problem. It was shown by Blum~\etal{} \cite{blumetal,blumetalOR} that a $(1/2 - \epsilon)$-approximate matching of $G_p$ can be found by querying $O_{p, \epsilon}(1)$ edges of each vertex in $G$, where the dependence on $p$ was exponential. Later, Assadi, Khanna, and Li \cite{AKL16} improved the dependence on $p$ and obtained the same approximation with $\poly(1/\epsilon p)$ queries. Numerous follow up works \cite{AKL17,YM18,soda19,sosa19,stoc20,focs20} then improved the approximation ratio. Particularly, the algorithm of Assadi and Bernstein \cite{sosa19} (see also \cite{soda19}) obtains a $(2/3 - \epsilon)$-approximation with $\poly(1/\epsilon p)$ queries. It was already observed in \cite{AKL16} that $2/3$-approximation is a barrier for the problem. Recently, Behnezhad, Derakhshan, and Hajiaghayi \cite{stoc20} broke this barrier and showed that one can obtain a $(1-\epsilon)$-approximation with $O_{\epsilon, p}(1)$ per-vertex queries, where the dependence on $\epsilon$ and $p$ is super-polynomial. Determining the best approximation achievable via $\poly(1/\epsilon p)$ queries remains an important open question for the stochastic matching problem.

In light of this progress on the approximate matching problem, it is natural to ask whether the same can also be achieved for the dual minimum vertex cover problem. Particularly,

\begin{question}\label{question}
	Can we find an approximate MVC of $G_p$ by querying few, preferrably $\poly(1/p)$, edges of each vertex in the base graph $G$?
\end{question}

Observe that a vertex cover of $G$ is also a valid vertex cover of $G_p$ since $G_p \subseteq G$. However, since some of the edges of $G$ may not belong to $G_p$, the MVC of $G_p$ might be smaller than that of $G$. In fact, the MVC of $G$ may be as large as $\sfrac{1}{p}$ times the MVC of $G_p$ in expectation --- an example is when $G$ is simply a matching. It turns out that to obtain any constant approximation (independent of $p$), $\Omega(\sfrac{n}{p})$ total queries are necessary; see Theorem~\ref{theorem:oirwfno3iw}. One may wonder if randomly querying the edges of $G$ may help. However, we show that the set of queried edges must be picked with much more care --- see Section~\ref{sec:randomqueries} for why random queries do not work.

Perhaps the simplest known constant approximation for the MVC problem is through {\em maximal} matchings. The set of endpoints of the edges in a maximal matching is well-known to form a $2$-approximate MVC. For this to hold, however, the {\em maximality} of the matching is essential, and even a $(1-\epsilon)$-approximate maximum matching that is not maximal is not useful. Unfortunately, none of the works above on the stochastic matching problem yield a maximal matching of $G_p$. In fact, we prove a separation (Theorem~\ref{theorem:ourfiuer}) via a simple lower bound, that, unlike approximate matchings, finding a maximal matching of $G_p$ requires $\Omega(n \log n)$ total queries. The situation seems even more complicated on the algorithmic side. In fact, we do not know if a maximal matching of $G_p$ can be found with $o(n^2)$ queries (note that the whole graph $G_p$ can be learned with $O(n^2)$ queries).

\subsection{Our Results}

In this work, we make progress on Question~\ref{question} on several fronts. 

Our main end results are in the $\poly(1/p)$ query regime and we prove the following two results for general and bipartite graphs, respectively.

\begin{graytbox}
	\begin{result}[see Theorem~\ref{thm:generalvc}]\label{res:general}
		For any graph $G$, any $\epsilon > 0$, and any $p \in (0, 1]$, there is a poly-time algorithm that finds a $(2+\epsilon)$-approximate MVC of $G_p$ via $O(\frac{1}{\epsilon^3 p})$ per-vertex queries.
	\end{result}
\end{graytbox}

\begin{graytbox}
	\begin{result}[see Theorem~\ref{thm:bipartitevc}]\label{res:bipartitemain}
		For any bipartite graph $G$, and any $p \in (0, 1]$, there is a poly-time algorithm that finds a $1.367 (\approx \frac{e+1}{e})$-approximate MVC of $G_p$ via $\poly(1/p)$ per-vertex queries.
	\end{result}
\end{graytbox}

As discussed, $\Omega(\sfrac{1}{p})$ per-vertex queries are necessary to obtain any constant approximation; see Theorem~\ref{theorem:oirwfno3iw}. Therefore, Result~\ref{res:general} is asymptotically query-optimal. Note, on the other hand, that for $p = 1$, the problem reduces to the non-stochastic MVC problem. This means that, restricting ourselves to polynomial-time algorithms, the approximation ratio achieved in Result~\ref{res:general} is also optimal for general graphs (up to an additive $\epsilon$) under the Unique Games Conjecture \cite{DBLP:journals/jcss/KhotR08}.

For bipartite graphs, however, the UGC based lower bound does not hold. Indeed, an optimal MVC can be found in polynomial time. Result~\ref{res:bipartitemain} asserts that in the stochastic setting, too, one can get around the $2$-approximation barrier with just $\poly(1/p)$ per-vertex queries. 

To prove Result~\ref{res:bipartitemain} we prove a number of tools (overviewed in Section~\ref{sec:techniques}) that, in a sense, give a better understanding of matchings in stochastic graphs. Using these tools, we also obtain the guarantee of Corollary~\ref{cor:matching} for stochastic matching in bipartite graphs, improving the previous close to $2/3$ approximations of \cite{soda19,sosa19} (which respectively obtain $0.656$ and $(2/3-\epsilon)$-approximations) in the $\poly(1/p)$ query regime. We note that Corollary~\ref{cor:matching}, importantly, is the first result to break the $2/3$-approximation barrier of \cite{AKL16} with just $\poly(1/p)$ queries. 

\begin{corollary}\label{cor:matching} 
	For any bipartite graph $G$, and any $p \in (0, 1]$, there is a poly-time algorithm that finds a $0.731 (\approx \frac{e}{e+1})$-approximate matching of $G_p$ in expectation via $O(\frac{\log 1/p}{p})$ per-vertex queries.
\end{corollary}

Finally, we turn our attention to the regime where super-polynomial-in-$1/p$ queries per-vertex are allowed. We show that in this setting, the approximation guarantee of Result~\ref{res:bipartitemain} can be improved all the way to $(1+\epsilon)$.

\begin{result}[see Theorem~\ref{thm:bipartite1-eps}]\label{res:biparite1+eps}
	For any bipartite graph $G$, any $\epsilon > 0$, and any $p \in (0, 1]$, there is a poly-time algorithm that finds a $(1+\epsilon)$-approximate MVC of $G_p$ via $O_{\epsilon, p}(1)$ per-vertex queries.
\end{result}

We note that the dependence of the number of per-vertex queries in Result~\ref{res:biparite1+eps} on $p$ is in the order $\exp(\exp(\exp(\poly(1/p))))$. It remains an important open problem to determine whether a $(1+\epsilon)$-approximation for bipartite graphs is achievable via $\poly(\sfrac{1}{\epsilon p})$ queries. (The same is also open for stochastic matching as discussed.)

\section{Main Techniques}\label{sec:techniques}

All of our algorithms in this paper for the stochastic MVC problem return a subset $C \subseteq V$ which is with probability one a vertex cover of $G_p$. That is, all the edges of $G_p$ have at least one endpoint in $C$ at all times. Let $Q$ denote the subset of edges in $G$ that we query and let $S$ denote the rest of the edges.  Observe that since we are unaware of the realization of edges in $S$, we have to cover them all, no matter which ones are realized. Therefore, once we fix the subgraph $Q$ to be queried, the ``best algorithm'' is well-defined: Report a MVC of graph $H := Q_p \cup S$ which includes the realized edges in $Q$, but all the edges in $S$.\footnote{Since exact MVC is NP-hard for general graphs, we actually end up using a different algorithm for Result~\ref{res:general}.} Observe that since we require $Q$ to be sparse, the vast majority of edges will be in $S$ and are always assumed to be realized. The challenge is to ensure that the extra covering constraints imposed by these edges do not increase the size of our vertex cover by much, compared to the actual minimum vertex cover of $G_p$.

The discussion above actually unveils an interesting connection between the stochastic vertex cover problem and the stochastic matching problem particularly in bipartite graphs where by K\"onig's famous theorem MVC and maximum matching have the same size. On the one hand, the stochastic matching problem asks for a subgraph $Q$, such that if we \textbf{remove} all the rest of the edges $S := E \setminus Q$ from $G_p$, the size of maximum matching in graph $G_p \setminus S$ remains close to that of $G_p$. On the other hand, the stochastic vertex cover problem asks for a subgraph $Q$ such that if we \textbf{add} all the edges $S := E \setminus Q$ to $G_p$, the size of the maximum matching (which equals the size of MVC) in $G_p \cup S$ remains close to that of $G_p$.

Now let us describe how we actually pick subgraph $Q$ to query, and how we analyze the size of the minimum vertex cover achieved by querying this subgraph. 

\smparagraph{The Half-Stochastic Matching Lemma:} This lemma, which we prove in Section~\ref{sec:HSML}, is one of the main components of our paper. It mainly provides a partitioning $Q, S$ of the edge-set $E$ of $G$. Let $H := Q_p \cup S$ denote a ``half-stochastic'' graph which includes each edge of $Q$ independently with probability $p$ but includes all the edges of $S$ with probability 1. We use this partitioning in our algorithm for Result~\ref{res:bipartitemain} in particular. There, we query only the edges in $Q$ and report the MVC of graph $H$, as outlined before. As a result, the partitioning should clearly ensure $Q$ has a small maximum degree. In addition, the nice property of this partitioning is that the edges $e \in S$ have a relatively small probability $\leq \epsilon^2 p$ of being part of matching $\mc{M}(H)$, where $\mc{M}$ is a (near) maximum matching algorithm that is also provided by the lemma. Next, we describe some of the challenges that we face in proving this lemma and how we overcome them.

Let us start with the trivial partitioning $Q_0 = \emptyset, S_0 = E$ and let $\mc{M}$ be an arbitrary (possibly randomized) maximum matching algorithm. The problem with this solution is that an edge $e \in S_0$ may have a large probability of being part of $\mc{M}(H_0)$. This occurs if some edges in $H_0$ are crucial for the matching to be maximum. We can try to put these edges of $S_0$ violating the probability constraint in $Q_0$, and obtain a new partitioning $Q_1, S_1$. The problem, however, is that the corresponding graph $H_1$ has a different distribution than graph $H_0$. Thus, it could be that edges in $S_1$ that previously had a small probability of joining $\mc{M}(H_0)$ now become crucial for the maximum matching of $H_1$. This can in fact continue for a super-constant number of iterations, inevitably violating the maximum degree constraint of $Q$. 

Instead of using an arbitrary maximum matching algorithm, our first idea is to use a special matching algorithm $\mc{M}$ that maximizes the following objective:
$$
	\Phi := \sum_{e \in E} \left(\Pr[e \in \mc{M}(H)] - \epsilon \Pr[e \in \mc{M}(H)]^2\right).
$$
The first term in the sum, intuitively, ensures that the size of matchings produced by the algorithm is large. The second term, intuitively, is to ensure the edges  tend to have small probabilities of joining $\mc{M}(H)$. A nice ``averaging'' property of this objective, is that if some matching algorithm $\mc{M}_1$ guarantees an objective of $\Phi_1$ and another algorithm $\mc{M}_2$ guarantees $\Phi_2$, then the algorithm that with probability $1/2$ picks the output of $\mc{M}_1$ and with probability $1/2$ the output of $\mc{M}_2$, has objective strictly larger than $\frac{\Phi_1 + \Phi_2}{2}$, unless the vast majority of edges have the same probability of joining $\mc{M}_1$ and $\mc{M}_2$.  

We plug in this new matching algorithm in the aforementioned framework for obtaining a list of partitionings $(Q_0, S_0), \ldots, (Q_k, S_k)$. But now, we use the averaging property of our special matching algorithm to argue that we reach our desired partitioning for some $k = \poly(1/\epsilon p)$.

\smparagraph{A New Vertex-Independent Matching Algorithm:}
 Now suppose that we have the partitioning $(Q, S)$ provided by the lemma discussed above. How should we argue that the MVC of graph $H = Q_p \cup S$ approximates the MVC of the actual realization $G_p$? One of the key parts of our analysis, which we discuss thoroughly in Section~\ref{sec:vertexindependent}, is a new {\em vertex-independent matching} (VIM) lemma. VIMs were introduced recently in \cite{stoc20} (further refined in \cite{focs20}) and were shown to be extremely useful for stochastic matchings. Roughly speaking, given a stochastic graph $G_p$ and an arbitrary matching algorithm $\mc{A}$, a VIM algorithm $\mc{B}$ has the property that its matching $\mc{B}(G_p)$ approximates $\mc{A}(G_p)$ (in both the total size and marginal probabilities of edges/vertices joining the matching), but in addition, for ``most'' vertices $u$ and $v$ of the graph, whether or not they are matched in $\mc{B}(G_p)$ are independent events. In the works of \cite{stoc20,focs20} for example, this independence is satisfied for $u, v \in V$ if they are at distance at least $\polylog \Delta$ in $G$, where $\Delta$ is the maximum degree of $G$. This requirement on the distance to achieve independence is provably necessary for the approach taken in \cite{stoc20,focs20} which is through distributed local algorithms. In contrast, we use a completely different approach to achieve independence in this work. Our new VIM works for bipartite graphs, but unlike prior works, for any two vertices $u$ and $v$ (in different partitions) that are non-adjacent, we have independence. This independence in particular holds, even if $u$ and $v$ are connected via a path of length 3. This better guarantee on independence is the key, for example, to why we are allowed to break the RS-barrier for stochastic matchings via $\poly(1/p)$ queries, whereas the previous approaches \cite{stoc20,focs20} required a super-polynomial in $1/p$ queries. It is also used, crucially, in the analysis of the stochastic MVC algorithm we described above for Result~\ref{res:bipartitemain}.

\section{Preliminaries \& Paper Organization}

\smparagraph{Notation.} For any graph $G$, we use $\nu(G)$ to denote the size of the minimum vertex cover of $G$ and use $\mu(G)$ to denote the size of the maximum matching of $G$. A ``fractional matching'' $\b{x}$ of a graph $G=(V, E)$ is an assignment $\{x_e\}_{e \in E}$ to the edges, where $x_e \in [0, 1]$ and for each vertex $v \in V$, $x_v := \sum_{e \ni v} x_e \leq 1$. We use $|\b{x}| := \sum_e x_e$ to denote the size of a fractional matching and for any subset $E' \subseteq E$, use $\b{x}(E')$ to denote $\sum_{e \in E'} x_e$. For any integer $k$ we use $[k]$ to denote set $\{1, \ldots, k\}$. We say $S_1, \ldots, S_k$ ``partitions'' set $S$ if $S_1 \cup \ldots \cup S_k = S$ and $S_i \cap S_j = \emptyset$ for all $i, j \in [k]$.

As in the literature (see e.g. \cite{blumetal,stoc20}), we say a (random) matching $M$ provides an $0 < \alpha \leq 1$ approximation for the stochastic matching problem if $M \subseteq G_p$ and $\E|M| \geq \alpha \cdot \E[\mu(G_p)]$. We say a (random) subset $C \subseteq V$ is a $\beta \geq 1$ approximate stochastic minimum vertex cover, if any edge in $G_p$ has at least an endpoint in $C$ with probability 1, and $\E|C| \leq \beta \cdot \E[\nu(G_p)]$.

We use the following well-known propositions throughout the paper.

\begin{proposition}[K\"onig's Theorem]\label{prop:konig}
	In any bipartite graph $G$, $\nu(G) = \mu(G)$.
\end{proposition}

\begin{proposition}[Chebyshev's inequality]
	Let $X$ be a random variable with finite expected value $\E[X]$ and finite non-zero variance $\Var[X]$. For any $\lambda > 0$,  
	$$\Pr[|X-\E X| \geq \lambda] \leq \frac{\Var[X]}{\lambda^2}.$$
\end{proposition}

\smparagraph{Paper organization.} In Sections~\ref{sec:HSML} and \ref{sec:vertexindependent} we prove two of the main tools introduced in this paper, particularly the ``half-stochastic matching lemma'' and the ``new vertex-independent matching lemma''. In Section~\ref{sec:upperbounds} we present our algorithms for the stochastic MVC problem and also the  improved result for the stochastic matching problem. Finally, in Section~\ref{sec:lowerbounds} we prove several lower bounds for both the stochastic vertex cover problem and the stochastic matching problem. 

\section{Tool I: The Half-Stochastic Matching Lemma}\label{sec:HSML}

The Half-Stochastic Matching Lemma, constructively, gives a partitioning $(Q, S)$ of the edge-set $E$ of graph $G$. This partitioning is accompanied with a special near-maximum matching algorithm $\mc{M}$ that operates on the ``half-stochastic'' random subgraph $H$ of $G$ in which each edge of $Q$ is stochastic (i.e. realized with probability $p$) and each edge of $S$ appears with probability one.  

Although this lemma seems to be about matchings only, it actually plays an important role in the stochastic vertex cover algorithm for Theorem~\ref{thm:bipartitevc} in both the algorithm in deciding which edges to query, and the analysis of the approximation ratio achieved by this algorithm. 

\begin{highlighttechnical}
\begin{lemma}[Half-Stochastic Matching Lemma]\label{lem:vcpartition}
\parindent=15pt
	Let $G=(V, E)$ be a (possibly non-bipartite) graph and let $\epsilon \in [0, 1]$ and $p \in [0, 1]$ be two given parameters. There is a partitioning of $E$ into  subsets $Q$ and $S$ such that:
\begin{enumerate}[itemsep=2pt, topsep=10pt,label={$(\roman*)$}]
		\item The maximum degree in $Q$ is $O(\frac{1}{\epsilon^{11} p^6})$.\label{vp-prop:degreeQ}
\end{enumerate}
Moreover, $Q$ and $S$ are such that there exists a randomized matching algorithm $\mc{M}$, where by letting $H := Q_p \cup S$ denote a random graph which includes all edges in $S$, but includes each edge of $Q$ independently with probability $p$, we get:
\begin{enumerate}[itemsep=2pt, topsep=10pt, label={$(\roman*)$}]
\setcounter{enumi}{1}
		\item $\E|\mc{M}(H)| \geq (1-2\epsilon) \cdot \E[\mu(H)]$. That is, the matching algorithm $\mc{M}$ should find an approximate maximum matching of $H$ in expectation. \label{vp-prop:expmatching}
		\item For any edge $e \in S$, $\Pr[e \in \mc{M}(H)] \leq \Stau$.\label{vp-prop:prsmallS}
	\end{enumerate}
	
	We note that the probabilistic statements above are with respect to both the inherent randomization in graph $H$, and also the randomization used in algorithm $\mc{M}$.
\end{lemma}
\end{highlighttechnical}

\newcommand{\threshD}[0]{\ensuremath{\epsilon^4 p}}
\newcommand{\SM}[1]{\ensuremath{\mc{M}_{#1}}}

We now turn to present the proof of Lemma~\ref{lem:vcpartition}. To do so, we have to output a triplet $(Q, S, \mc{M})$ where $Q$ and $S$ form a partitioning of $E$, and $\mc{M}$ is a matching algorithm. 

For some $k = O_{\epsilon, p}(1)$ which we specify later, we construct a list $(Q_0, S_0), (Q_1, S_1), \ldots, (Q_k, S_k)$ of partitionings. We will argue that if $k$ is large enough, there is one of the partitions, $(Q_i, S_i)$, which satisfies all the properties required by Lemma~\ref{lem:vcpartition}. This will be our partitioning $(Q, S)$.

\smparagraph{The construction.} The base partitioning is simply $Q_0 = \emptyset, S_0 = E$ and for each $i \geq 0$,  $(Q_{i+1}, S_{i+1})$ is constructed from the previous partitioning $(Q_i, S_i)$. To describe the construction, let us define for each partitioning $i$ a random graph $H_i$ which includes each edge of $Q_i$ independently with probability $p$, and includes every edge of $S_i$ with probability 1. We will also soon formalize a special randomized matching algorithm $\SM{i}$ that we run on graph $H_i$. Having $\SM{i}$, for each $i \geq 0$ we define
\begin{equation}\label{eq:defD}
	D_i := \{ e \in S_i \mid \Pr[e \in \SM{i}(H_i)] > \Stau \},
\end{equation}
and construct $(Q_{i+1}, S_{i+1})$ in the following way:
\begin{equation}\label{eq:constr-partitn}
	Q_{i+1} \gets Q_i \cup D_i \qquad \text{and} \qquad S_{i+1} \gets S_i \setminus D_i.
\end{equation}

Observe that in the construction above, each partitioning $(Q_{i+1}, S_{i+1})$ is obtained from the previous partitioning $(Q_{i}, S_{i})$ by ``moving'' the edges of $D_i$ (which are all by definition in $S_i$) from $S_i$ to partition $Q_i$. For this reason, we have
$$
	Q_0 \subset Q_1 \subset \ldots \subset Q_k \qquad \text{and} \qquad S_0 \supset S_1 \supset \ldots \supset S_k.
$$

Let us emphasize that graphs $H_0, H_1, \ldots, H_k$ have \textbf{different} distributions. Intuitively, since $Q_i$ grows as $i$ increases, and graph $H_i$ includes only $p$ fraction of the edges in $Q_i$ but all the edges in $S_i$, random graph $H_i$ tends to get smaller and smaller by increasing $i$. In fact it would be useful to consider a coupling $(H_0, H_1, \ldots, H_k)$ as follows: On each edge $e$ we draw an independent $p$-Bernouli random variable $X_e$ and use this to define $H_i$ for all $i$ as:
\begin{equation}\label{eq:coupling}
	H_i = \{e \mid (e \in Q_i \text{ and } X_e = 1) \text{ or } (e \in S_i)\}.
\end{equation}
This coupling is useful because in each outcome of the joint distribution $(H_0, \ldots, H_k)$, each graph $H_i$ is a subgraph of the previous graph $H_{i-1}$.

Let us now finalize our construction by formalizing algorithm $\SM{i}$.

\smparagraph{The Matching $\SM{i}$:} For any $i$ and any (possibly randomized) matching algorithm $\mc{M}'$, we define
\begin{equation}\label{eq:objective}
	\Phi_i(\mc{M}') := \sum_{e \in E} \left(\Pr_{\mc{M}', H_i}[e \in \mc{M}'(H_i)] - \epsilon \Pr_{\mc{M}', H_i}[e \in \mc{M}'(H_i)]^2\right),
\end{equation}
where let us emphasize that the probabilities are taken over both the possible randomization in algorithm $\mc{M}'$, and the randomization in graph $H_i$ (regarding the realization of edges belonging to $Q_i$). Having this definition, we now simply let $\SM{i}$ be the algorithm maximizing $\Phi_i$, i.e.:
$$
	\SM{i} := \argmax_{\mc{M}'} \Phi_i(\mc{M}'),
$$
and we use $\Phi_i := \Phi_i(\SM{i})$ to simply denote the optimal objective value obtained by this algorithm.

Let us, for now, not concern ourselves with how the objective function (\ref{eq:objective}) can be maximized in polynomial time and assume that this algorithm $\SM{i}$ is simply given. We will later address this issue and obtain a polynomial-time algorithm once it becomes clear how we use $\SM{i}$.

\smparagraph{The intuition behind objective $(\ref{eq:objective})$.} To see the intuition behind why we define the objective function (\ref{eq:objective}) this way, first note by linearity of expectation that
\begin{equation}\label{eq:sn123908}
	\Phi_i = \E|\SM{i}(H_i)| - \epsilon \sum_{e \in E} \Pr[e \in \SM{i}(H_i)]^2.
\end{equation}
The intuition behind the first term is clear: We want the expected size of the matching to be large;  Observation~\ref{obs:matchingSMlarge} below formalizes this by showing that the matching algorithm $\SM{i}$ maximizing $\Phi_i$ must be a $(1-\epsilon)$-approximate matching. The second term, on the other hand, ensures that the marginal probabilities of edges appearing in the matching tend to be small. This is useful for Property~\ref{vp-prop:prsmallS} of Lemma~\ref{lem:vcpartition} which requires small marginals for all edges in $S$.

\begin{observation}\label{obs:EM>Phi}
	For any $i$, $\E|\SM{i}(H_i)| \geq \Phi_i$.
\end{observation}
\begin{proof}
	By Equation (\ref{eq:sn123908}), $\E|\SM{i}(H_i)| = \Phi_i + \epsilon \sum_{e \in E} \Pr[e \in \SM{i}(H_i)]^2 \geq \Phi_i$.
\end{proof}

\begin{observation}\label{obs:matchingSMlarge}
	$\E|\SM{i}(H_i)| \geq (1-\epsilon) \cdot \E[\mu(H_i)]$.
\end{observation}
\begin{proof}
	Consider a deterministic algorithm $\mc{M}'$ that picks an arbitrary maximum matching $M$ of its random input $H_i$. Since $\E|M| = \E[\mu(H_i)]$, we have
	$$
	\Phi_i(\mc{M}') = \E[\mu(H_i)] - \epsilon \sum_{e \in E} \Pr[e \in M]^2 \geq \E[\mu(H_i)] - \epsilon \sum_{e \in E} \Pr[e \in M] = (1-\epsilon) \cdot \E[\mu(H_i)].
	$$
	Since $\SM{i}$ maximizes $\Phi_i = \Phi_i(\SM{i})$, we get $\Phi_i \geq \Phi_i(\mc{M}') \geq (1-\epsilon) \cdot \E[\mu(H_i)]$. Combined with $\E|\SM{i}(H_i)| \geq \Phi_i$ due to Observation~\ref{obs:EM>Phi}, this implies $\E|\SM{i}(H_i)| \geq (1-\epsilon) \cdot \E[\mu(H_i)]$.
\end{proof}

We now turn to prove that one of the partitionings $(Q_i, S_i)$ must satisfy the properties required by Lemma~\ref{lem:vcpartition}. The next set of claims are used for this purpose.

\begin{claim}\label{cl:matchingsdecreasing}
	It holds that
	$$
		\mu(G) \geq \Phi_0 \geq \Phi_1 \geq \ldots \geq \Phi_k \geq (1-\epsilon) \cdot p \cdot \mu(G).
	$$
\end{claim}

\begin{claim}\label{cl:interval}
	There is an interval $I = \{s, \ldots, s+k'\}$ in $[k]$ with $|I| \geq \frac{1}{300} \epsilon^6 p^3 k$ such that 
	\begin{flalign*}
		|\Phi_i - \Phi_j| \leq 0.01 \epsilon^6 p^3 \mu(G) \qquad\qquad \text{for all $i < j$ in $I$.}
	\end{flalign*}
\end{claim}

\begin{claim}\label{cl:Dicontinuestocontribute}
	Let $I$ be as defined in Claim~\ref{cl:interval}. Either there is some $i \in I$ where $\E|D_i \cap \SM{i}(H_i)| < \epsilon p \mu(G)$, or otherwise for any $i, j \in I$ with $i < j$, it holds that
	$$
		\E|D_i \cap \SM{j}(H_j)| \geq 0.25 \epsilon^3 p^2 \cdot  \mu(G).
	$$
\end{claim}

Claims~\ref{cl:matchingsdecreasing} and \ref{cl:interval} are proved in Section~\ref{sec:proofs-uh123} and Claim~\ref{cl:Dicontinuestocontribute} is proved in Section~\ref{sec:proof-dgl1298371923}. Claim~\ref{cl:Dicontinuestocontribute} is, in particular, the key part of the proof. It is proved by showing that if the condition of Claim~\ref{cl:Dicontinuestocontribute} is not satisfied, then the randomized matching algorithm that with probability $0.5$ picks the output of $\SM{j}(H_j)$ and otherwise the output of $\SM{i}(H_i)$, should obtain a larger objective than $\Phi_i$ which we show is a contradiction. 

Having proved these claims, we now turn to prove Lemma~\ref{lem:vcpartition}. 

\begin{proof}[Proof of Lemma~\ref{lem:vcpartition}]
First, we set $k = 5 \cdot \frac{300}{\epsilon^9 p^5}$. Since for any $i$, each edge in $D_i$ has probability at least $\epsilon^2 p$ of being in matching $\SM{i}(H_i)$, and that the probabilities around each vertex sum up to at most one, there are at most $\sfrac{1}{\epsilon^2 p}$ edges connected to each vertex in $D_i$. This implies that for any $j \in [k]$, $Q_j$ has maximum degree at most $k \cdot \frac{1}{\epsilon^2 p} = O(\frac{1}{\epsilon^{11} p^6})$, satisfying Property~\ref{vp-prop:degreeQ}. Now we prove there exists some $(Q_i, S_i)$ satisfying Properties~\ref{vp-prop:expmatching} and \ref{vp-prop:prsmallS} as well.

Let $I$ be as provided by Claim~\ref{cl:interval}. There are two possible cases:

	\smparagraph{Case 1 ---} There is some $i \in I$ where $\E|D_i \cap \SM{i}(H_i)| < \epsilon p \mu(G)$\textbf{:} 
	
	\noindent In this case, we can let $Q \gets Q_i, S \gets S_i,$ which implies graph $H$ of Lemma~\ref{lem:vcpartition} has the same distribution as $H_i$. We now let matching algorithm $\mc{M}$, required by Lemma~\ref{lem:vcpartition}, to be the same as matching algorithm $\SM{i}$, except that we exclude the edges of $D_i$ from the matching. That is, we let $\mc{M}(H) = \SM{i}(H_i) \setminus D_i$.

	Since we exclude the edges in $D_i$ from the matching,  we get that for all edges $e \in S$, $\Pr[e \in \mc{M}(H)] \leq \Stau$ satisfying Property~\ref{vp-prop:prsmallS}. On the other hand, 
	\begin{flalign*}
		\E|\mc{M}(H)| &= \E[ |\SM{i}(H_i) \setminus  D_i| ] = \E|\SM{i}(H_i)| - \E|D_i \cap \SM{i}(H_i)|\\
		&> \E|\SM{i}(H_i)| - \epsilon p \mu(G) \tag{By the assumption of Case 1.}\\
		&\geq (1-\epsilon) \E[\mu(H_i)] - \epsilon p \mu(G) \tag{By Observation~\ref{obs:matchingSMlarge}.}\\
		&\geq (1-2\epsilon) \E[\mu(H_i)] \tag{Since $\E[\mu(H_i)] \geq p\mu(G)$.}\\
		&= (1-2\epsilon) \E[\mu(H)]. \tag{Since $H$ and $H_i$ have the same distribution.}
	\end{flalign*}
	This proves $\mc{M}$ is a $(1-2\epsilon)$-approximate matching algorithm, satisfying Property~\ref{vp-prop:expmatching}.

	\smparagraph{Case 2 ---} For all $i \in I$, $\E|D_i \cap \SM{i}(H_i)| \geq \epsilon p \mu(G)$\textbf{:}
	
	\noindent In this case, by Claim~\ref{cl:Dicontinuestocontribute}, we have $\E|D_i \cap \SM{j}(H_j)| \geq 0.25 \epsilon^3 p^2 \cdot  \mu(G)$ for all $i < j$ in $I$. Let us denote $I = \{a_1, a_2, \ldots, a_{\ell}\}$ where $a_1 < \ldots < a_{\ell}$. Letting $j = a_{\ell}$, we thus get
	$$
		\E|D_{a_i} \cap \SM{a_\ell}(H_{a_\ell})| \geq 0.25 \epsilon^3 p^2 \cdot  \mu(G) \qquad\text{for all $i \in \{1, \ldots, \ell - 1\}$}.
	$$
	On the other hand, observe from construction (\ref{eq:constr-partitn}) that sets $D_{a_1}, \ldots, D_{a_{\ell - 1}}$ are all pairwise disjoint. This implies
	\begin{equation}\label{eq:bheug123981273}
		\E|\SM{a_\ell}(H_{a_\ell})| \geq \sum_{i=1}^{\ell - 1} \E|D_{a_i} \cap \SM{a_\ell}(H_{a_\ell})| \geq (\ell - 1) \cdot 0.25\epsilon^3 p^2 \cdot \mu(G).
	\end{equation}
	Recall from Claim~\ref{cl:interval} that $\ell = |I| \geq \frac{1}{300} \epsilon^6 p^3 k$. Since we set  $k = 5 \cdot \frac{300}{\epsilon^9 p^5}$, we get $\ell \geq \frac{5}{\epsilon^3 p^2}$ which combined with (\ref{eq:bheug123981273}) implies $\E|\SM{a_\ell}(H_{a_\ell})| > \mu(G)$ which is a contradiction since $H_{a_\ell}$ is a subgraph of $G$ and cannot have a larger matching than $\mu(G)$. This contradiction implies that if we set $k$ large enough, this second case essentially does not happen. As a result, we always end up at Case 1, which we just showed how it proves Lemma~\ref{lem:vcpartition}.

The proof of Lemma~\ref{lem:vcpartition} is thus complete.
\end{proof}

Finally, we remark that our techniques also lead to a partitioning with the same guarantee as in  Lemma~\ref{lem:vcpartition} that can be found in polynomial time. We defer the details of this polynomial-time implementation to Appnedix~\ref{apx:polytime}.

\subsection{Proofs of Claims~\ref{cl:matchingsdecreasing} and \ref{cl:interval}}\label{sec:proofs-uh123}
\begin{proof}[Proof of Claim~\ref{cl:matchingsdecreasing}]
	As discussed, in the coupling of Eq~\ref{eq:coupling}, $H_i$ is always a subgraph of $H_{i-1}$. As a result, matching algorithm $\SM{i}$ is also applicable on graph $H_{i-1}$, implying that $\Phi_{i-1} \geq \Phi_i$ for all $i$.
	
	To see why $\mu(G) \geq \Phi_0$, note from Observation~\ref{obs:EM>Phi} that $\E|\SM{0}(H_0)| \geq \Phi_0$. On the other hand, no matter what matching algorithm we use for $\SM{0}$, we have $\E|\SM{0}(H_0)| \leq \mu(G)$ as the output must be a matching in $H_0$ and thus $G$. Combining the two bounds gives $\mu(G) \geq \Phi_0$.
	
	Finally, to see why $\Phi_k \geq (1-\epsilon) \cdot p\cdot \mu(G)$, fix a maximum matching $M$ of $G$ which has to have size $\mu(G)$. Every edge $e \in M$ either is in $Q_k$ or $S_k$; in either case, $e \in H_k$ with probability at least $p$. We thus have $\E[\mu(H_k)] \geq \E|M \cap H_k| \geq p |M| = p \cdot \mu(G)$. Now consider a choice for $\SM{k}$ which deterministically picks a maximum matching $M_k$ of $H_k$. This proves
	\begin{flalign*}
		\Phi_k &\geq \sum_{e \in E} \Pr[e \in M_k] - \epsilon \Pr[e \in M_k]^2 \geq \sum_{e \in E} \Pr[e \in M_k] - \epsilon \Pr[e \in M_k]\\
		&= (1-\epsilon) \sum_{e \in E} \Pr[e \in M_k] 
		\geq (1-\epsilon) \cdot \E|M_k|
		= (1-\epsilon) \cdot \E[\mu(H_k)] 
		\geq (1-\epsilon) \cdot p \cdot \mu(G),
	\end{flalign*}
	completing the proof.
\end{proof}

\begin{proof}[Proof of Claim~\ref{cl:interval}]
	Let us define $I_j$ for any $j \in \{1, \ldots, \lceil 100/\epsilon^6 p^3\rceil\}$ as follows
	$$
		I_j := \{ i : (1- 0.01 j \epsilon^6 p^3) \mu(G) < \Phi_i \leq (1- 0.01 (j-1) \epsilon^6 p^3) \mu(G) \}.
	$$
	Recall from Claim~\ref{cl:matchingsdecreasing} that 
	$
		\mu(G) = \Phi_0 \geq \Phi_1 \geq \ldots \geq \Phi_k \geq (1-\epsilon)p\mu(G).
	$ 
	Thus, $I_1, I_2, \ldots$ partition $[k]$ into consecutive intervals where for all elements $i, j$ in the same interval, $|\Phi_i - \Phi_j| \leq 0.01 \epsilon^6 p^3 \mu(G)$.
	
	Since there are only $\lceil 100/ \epsilon^6 p^3 \rceil$ intervals and $\sum_j |I_j| = k$ (as every $i \in [k]$ belongs to exactly one of the intervals) there is at least one interval $I$ with $|I| \geq \frac{k}{\lceil 100/ \epsilon^6  p^3 \rceil + 1} > \frac{1}{300} \cdot \epsilon^6 p^3 k$. This interval $I$ satisfies the required property of the claim by its definition, and has the desired size.
\end{proof}

\subsection{Proof of Claim~\ref{cl:Dicontinuestocontribute}}\label{sec:proof-dgl1298371923}

\begin{proof}[Proof of Claim~\ref{cl:Dicontinuestocontribute}]
	Suppose for the sake of contradiction that 
	\begin{equation}\label{eq:gfc12839}
		\E|D_i \cap \SM{i}(H_i)| \geq \epsilon p \mu(G) \qquad \qquad \text{for all $i \in I$},
	\end{equation}
	and that there are $i, j \in I$ such that $i < j$ and 
	\begin{equation}\label{eq:bue129387}
		\E|D_i \cap \SM{j}(H_j)| <  0.25 \epsilon^3 p^2 \cdot \mu(G).
	\end{equation}
	Consider a matching $\SM{i, j}(H_i)$ which with probability $1/2$ returns the output of $\SM{i}(H_i)$ and with probability $1/2$ returns the output of $\SM{j}(H_j)$. Since $i < j$, by the coupling (\ref{eq:coupling}), $H_j$ is a subgraph of $H_i$ and thus any matching in $H_j$ is a matching in $H_i$. As a result, $\SM{i, j}$ is a valid matching algorithm for graph $H_i$. We prove that under (\ref{eq:gfc12839}) and (\ref{eq:bue129387}), algorithm $\SM{i, j}$ should satisfy $\Phi_i(\SM{i, j}) > \Phi_i(\SM{i})$ which contradicts the assumption that $\SM{i}$ maximizes $\Phi_i(\SM{i})$. 
	
	From the definition of objective $\Phi_i(\SM{i, j})$ we have
	\begin{flalign*}
		\Phi_i(\SM{i, j}) &= \sum_{e \in E} \Pr[e \in \SM{i, j}(H_i)] - \epsilon \Pr[e \in \SM{i, j}(H_i)]^2\\
		&= \sum_{e \in E} \left(\frac{\Pr[e \in \SM{i}(H_i)]+\Pr[e \in \SM{j}(H_j)]}{2}\right) - \epsilon \left(\frac{\Pr[e \in \SM{i}(H_i)]+\Pr[e \in \SM{j}(H_j)]}{2}\right)^2.
	\end{flalign*}
	Let us for simplicity of notation use $p_i(e) := \Pr[e \in \SM{i}(H_i)]$ and $p_j(e) := \Pr[e \in \SM{j}(H_j)]$. The equality above therefore can be expressed as
	\begin{equation}\label{eq:ce219837}
		\Phi_i(\SM{i, j}) = \sum_{e \in E} \left(\frac{p_i(e) + p_j(e)}{2} \right) - \epsilon \left(\frac{p_i(e) + p_j(e)}{2}\right)^2.
	\end{equation}
	Basic mathematical calculations give that for any $e$,
	\begin{flalign}
		\nonumber \left(\frac{p_i(e) + p_j(e)}{2}\right)^2 &= \frac{p_i(e)^2 + p_j(e)^2 +2p_i(e)p_j(e)}{4}\\
		\nonumber &= \frac{p_i(e)^2}{2} + \frac{p_j(e)^2}{2} - \frac{p_i(e)^2}{4} - \frac{p_j(e)^2}{4} + \frac{2p_i(e)p_j(e)}{4}\\
		&= \frac{p_i(e)^2}{2} + \frac{p_j(e)^2}{2} - \left( \frac{p_i(e) - p_j(e)}{2} \right)^2.\label{eq:gc128937123}
	\end{flalign}
	Replacing (\ref{eq:gc128937123}) back into (\ref{eq:ce219837}) gives
	\begin{flalign*}
		\Phi_i(\SM{i, j}) &= \sum_{e \in E} \left(\frac{p_i(e) + p_j(e)}{2} \right) - \epsilon \left( \frac{p_i(e)^2}{2} + \frac{p_j(e)^2}{2} - \left( \frac{p_i(e) - p_j(e)}{2} \right)^2 \right) \\
		&= \sum_{e \in E} \left(\frac{p_i(e) - \epsilon p_i(e)^2}{2}  \right) + \left(\frac{p_j(e) - \epsilon p_j(e)^2}{2}  \right) + \epsilon \left( \frac{p_i(e) - p_j(e)}{2} \right)^2.
		\tag{By simply moving the terms in the previous line.}\\
		&= \frac{\Phi_i}{2} + \frac{\Phi_j}{2} + \epsilon \sum_{e \in E} \left( \frac{p_i(e) - p_j(e)}{2} \right)^2 \tag{See below.}\\
		&\geq \Phi_i - 0.01\epsilon^6 p^3 \mu(G) + \epsilon \sum_{e \in E} \left( \frac{p_i(e) - p_j(e)}{2} \right)^2 \tag{Since $|\Phi_i - \Phi_j| \leq 0.01 \epsilon^6 p^3 \mu(G)$ by Claim~\ref{cl:interval}.}
	\end{flalign*}
	The third equality above, simply comes from the definition (\ref{eq:objective}) for $\Phi_i$, which implies $\Phi_i = \sum_e \Pr[e \in \SM{i}(H_i)] - \epsilon \Pr[e \in \SM{i}(H_i)]^2 = \sum_e p_i(e) - \epsilon p_i(e)^2$, and from the same bound applied on $\Phi_j$.
	
	Now define subset $D'_i := \{ e \in D_i \mid p_j(e) < 0.5 \epsilon^2 p \}$ of $D_i$. Using this subset only instead of the set $E$ of edges in the inequality above gives
	\begin{flalign*}
		\Phi_i(\SM{i, j}) &\geq \Phi_i - 0.01\epsilon^6 p^3 \mu(G)  + \epsilon \sum_{e \in D'_i} \left( \frac{p_i(e) - p_j(e)}{2} \right)^2 \\
		&\geq \Phi_i - 0.01\epsilon^6 p^3\mu(G) + \epsilon \sum_{e \in D'_i} \left(\frac{\epsilon^2 p - 0.5 \epsilon^2 p}{2}\right)^2 \tag{Since for any $e \in D'_i$, $p_j(e) < 0.5\epsilon^2 p$ by definition of $D'_i$ and $p_i(e) \geq \epsilon^2 p$ by definition of $D_i$.}\\
		&= \Phi_i - 0.01\epsilon^6 p^3 \mu(G) + \epsilon \sum_{e \in D'_i} \frac{\epsilon^4 p^2}{16}\\
		&= \Phi_i - 0.01\epsilon^6 p^3 \mu(G)  + \frac{\epsilon^5 p^2}{16} |D'_i|.
	\end{flalign*}
	To obtain the claimed contradiction, we will  prove that 
	\begin{equation}\label{eq:cgy182379123}
		|D'_i| \geq 0.5 \epsilon p \mu(G),
	\end{equation}
	which combined by inequality above proves
	\begin{equation}\label{eq:htcr1823973}
		\Phi_i(\SM{i, j}) \geq \Phi_i - 0.01\epsilon^6 p^3 \mu(G) + \frac{\epsilon^6 p^3}{32} \mu(G) >  \Phi_i + 0.01 \epsilon^6 p^3 \mu(G),
	\end{equation}
	which contradicts $\Phi_i$ being the maximum objective achievable.
	
	To complete the proof, it thus only remains to prove (\ref{eq:cgy182379123}). We have
	$$
		\E|D_i \cap \SM{j}(H_j)| = \sum_{e \in D_i} p_j(e)  = \sum_{e \in D_i \setminus D'_i} p_j(e) + \sum_{e \in D'_i} p_j(e) \geq \sum_{e \in D_i \setminus D'_i} p_j(e) \geq (|D_i| - |D'_i|) 0.5 \epsilon^2 p.
	$$
	Also note that 
	$$
		\epsilon p \mu(G) \stackrel{(\ref{eq:gfc12839})}{\leq} \E|D_i \cap \SM{i}(H_i)| = \sum_{e \in D_i} p_i(e) \leq \sum_{e \in D_i} 1 = |D_i|.
	$$
	Combining the two bounds above gives
	$$
		\E|D_i \cap \SM{j}(H_j)| \geq (\epsilon p \mu(G) - |D'_i|) 0.5 \epsilon^2 p = 0.5 \epsilon^3 p^2 \mu(G) - 0.5 \epsilon^2 p|D'_i|,
	$$
	which combined with the bound $\E|D_i \cap \SM{j}(H_j)| < 0.25 \epsilon^3 p^2 \cdot \mu(G)$ of (\ref{eq:bue129387}) gives
	$$
		0.25\epsilon^3 p^2 \mu(G) > 0.5 \epsilon^3 p^2 \mu(G) - 0.5 \epsilon^2 p |D'_i|.
	$$
	By moving the terms, we get
	$$
		|D'_i| > \frac{0.25 \epsilon^3 p^2 \mu(G)}{0.5 \epsilon^2 p} = 0.5 \epsilon p \mu(G).
	$$
	This is the desired bound of inequality (\ref{eq:cgy182379123}), which as discussed above completes the proof.
\end{proof}
 
\section{Tool II: A New Vertex-Independent Matching Lemma}\label{sec:vertexindependent}

The notion of ``vertex-independent matchings'' for stochastic graphs was introduced first in \cite{stoc20}. In this section, we present a new vertex-independent matching lemma for bipartite graphs, that unlike the previous ones \cite{stoc20,focs20}, which required the vertices to be far apart in the graph to have independence, guarantees independence for any pair of non-adjacent nodes, even if there is a short path of length 3 between them. This stronger guarantee on the independence is the key to improve per-vertex queries from $O_p(1)$ down to $\poly(1/p)$.

We use the vertex-independent lemma to prove the following which use for our bipartite MVC approximate algorithm.

\begin{lemma}\label{lem:appvil}
Let $G=(V, E)$ be a bipartite graph, let realization $G_p=(V, E_p)$ be a random subgraph of $G$ that includes each of its edges independently with probability $p$. Let $(Q, S)$ be a partitioning of $E$ and denote $Q_p := Q \cap E_p$ and $S_p := S \cap E_p$. Suppose also that we are given a (possibly randomized) matching algorithm $\mc{M}$, and a fractional matching $\b{q}$ on $E$ such that:
\begin{enumerate}
	\item For any edge $e \in Q$, $q_e = \Pr_{Q_p, \mc{M}}[e \in \mc{M}(Q_p)]$.
	\item For any edge $e \in S$, $q_e \leq \epsilon^5 p$ for some $\epsilon > 0$.
\end{enumerate}
Then $\E[\mu(G_p)] \geq (1-6\epsilon)\frac{e}{e+1}|\b{q}|$.
\end{lemma} 

We first present the vertex-independent matching algorithm in Section~\ref{sec:vialg} and use it to prove Lemma~\ref{lem:appvil} in Section~\ref{sec:proofofappvil}.

\subsection{The Vertex-Independent Matching Algorithm}\label{sec:vialg}

In this section we present our vertex-independent matching algorithm which satisfies the following:

\begin{lemma}[\textbf{Bipartite Vertex-Independent Matching Lemma}]\label{lem:vertexindependent}
Let $\Gamma(A, B, E_\Gamma)$ be a bipartite graph, let $\Gamma_p$ be a random subgraph of $\Gamma$ that contains any of its edges independently with probability $p$, let $\mc{A}$ be an \underline{arbitrary} matching algorithm, possibly randomized, and let $M_\mc{A}$ be the matching obtained by running $\mc{A}$ on $\Gamma_p$. There is a randomized algorithm (Algorithm~\ref{alg:B}) for constructing a matching $M_\mc{B}$ of $\Gamma_p$ such that:
\begin{enumerate}[label=$(\roman*)$]
\item $\E|M_{\mc{B}}| \geq (1-\frac{1}{e}) \cdot \E|M_{\mc{A}}|$.\label{viprop:expsize}
\item For any vertex $v \in A$, $\Pr[v \in M_\mc{B}] \leq \Pr[v \text{ proposes}] = \Pr[v \in M_\mc{A}]$.\\(See Algorithm~\ref{alg:B} for how the vertices on the $A$ side ``propose''.)\label{viprop:Asideprob}
\item For any vertex $u \in B$, $\Pr[u \in M_{\mc{B}}] \leq \Pr[u \in M_{\mc{A}}]$.\label{viprop:Bsideprob}
\item \label{viprop:independence} For any non-adjacent $v \in A, u \in B$ (i.e. $(u, v) \not\in E$), whether $v$ proposes (see Algorithm~\ref{alg:B}) is independent of event $u \in M_\mc{B}$.
\end{enumerate}
We emphasize that $M_{\mc{A}}$ (resp. $M_{\mc{B}}$) has two sources of randomization, one in the randomization of graph $\Gamma_p$, and one the possible randomization in algorithm $\mc{A}$ (resp. $\mc{B}$). The probabilistic statements above are with regards to both.
\end{lemma}

\begin{proof}
	We start by describing the algorithm for constructing matching $M_\mc{B}$. 
	
	For any vertex $v \in A$, let us use $R_v$ to denote the realization status of edges connected to $v$ in graph $\Gamma_p$. That is, $R_v$ reveals which edges connected to $v$ are realized and which ones are not realized, but crucially does not reveal any information about the realization of the rest of the edges. Using this information, for any edge $e = (v, u)$ with $v \in A, u \in B$, we define
	$
		p_e := \Pr[e \in M_\mc{A} \mid R_v].
	$
	That is, in defining $p_e$ for any edge $e = (v, u)$ we only need to know which edges connected to $v$ are realized in $\Gamma_p$, and are essentially unaware of realization of the rest of the edges in $\Gamma_p$. Similarly, for any vertex $v \in A$ we denote
	$
		p_v := \sum_{e \ni v} p_e = \Pr[v \in M_\mc{A} \mid R_v].
	$
	
	Observe that for any $v \in A$, $0 \leq p_v \leq 1$ since $p_v$ corresponds to the probability that $v$ is matched in $M_\mc{A}$ conditioned on the realization of its edges. Importantly, however, this does not hold for vertices of the other partition, and $\sum_{e \ni u} p_e$ may, in fact, exceed one for $u \in B$.
	
	Having defined $p_e$ and $p_v$ as above, the claimed Algorithm~\ref{alg:B} in Lemma~\ref{lem:vertexindependent} can be formalized:
	
\begin{tboxalg2e}{The algorithm for constructing matching $M_\mc{B}$ on  $\Gamma_p$.}
\begin{algorithm}[H]
	\DontPrintSemicolon
	\SetAlgoSkip{bigskip}
	\SetAlgoInsideSkip{}
	
	\label{alg:B}

	\For{any vertex $v\in A$}{
		Vertex $v$ either proposes to exactly one neighbor $u$, or does not propose at all. This is decided by a random procedure, where each neighbor $u$ has probability exactly $p_{(u, v)}$ of being proposed to by $v$, and there is a probability $1-p_v$ that $v$ does not propose.
	}

	\For{any vertex $u\in B$}{
	Among the vertices who sent proposals to vertex $u$, if any, $u$ chooses an arbitrary winner $v$ and we add $(u, v)$ to matching $M'_{\mc{B}}$.
	}
	
	\Return $M_\mc{B}$.
\end{algorithm}
\end{tboxalg2e}

We now prove the properties of Lemma~\ref{lem:vertexindependent}.

\smparagraph{Property~\ref{viprop:expsize}.}  Fix an arbitrary vertex $u \in B$, and let $\{v_1, \ldots, v_d\}$ be its neighbors in $A$. Let $Y_i$ be the indicator random variable for the event that $v_i$ proposes to $u$. First, observe that 
$$
	\E[Y_i] = \E[p_{(v_i, u)}] = \E[\Pr[(u, v_i) \in \mc{A}(\Gamma_p) \mid R_v] ] = \Pr[(u, v_i) \in \mc{A}(\Gamma_p) ].
$$
As a result,
\begin{equation}\label{eq:gr1238912388}
	\sum_{i} \E[Y_i] = \sum_{i} \Pr[(u, v_i) \in \mc{A}(\Gamma_p)] = \Pr[u \in \mc{A}(\Gamma_p)].
\end{equation}
Moreover, observe that $R_{v_1}, \ldots, R_{v_d}$ are mutually independent since the edges of $v_1, \ldots, v_d$ are all disjoint. Thus, the proposals of $v_1, \ldots, v_d$ are also mutually independent and so are random variables $Y_{1}, \ldots, Y_{d}$. Now observe that $u$ remains unmatched in $M_\mc{B}$ if and only if none of its neighbors proposes to it; combined with the independence discussed, this implies
$$
	\Pr[u \in M_\mc{B}] = 1 - \prod_{i} \Pr[Y_{i} = 0] = 1 - \prod_i (1-\E[Y_i]).
$$
Fixing the sum $\sum_i \E[Y_i]$ to be $S$, $1 - \prod_i (1-\E[Y_i])$ is minimized for $\E[Y_1] = \ldots = \E[Y_d] = \frac{S}{d}$. Thus
\begin{equation}\label{eq:ddgrl1283192837921}
\Pr[u \in M_\mc{B}] \geq 1 - \prod_{i=1}^d \left(1 - \frac{S}{d} \right) = 1 - \left(1-\frac{S}{d}\right)^d \geq \left( 1 - \sfrac{1}{e}\right) S \stackrel{(\ref{eq:gr1238912388})}{=} \left(1-\sfrac{1}{e}\right) \Pr[u \in \mc{A}(\Gamma_p)].
\end{equation}
Now, by linearity of expectation over all choices of $u \in B$, we get
$$
\E|M_\mc{B}| = \sum_{u \in B} \Pr[u \in M_\mc{B}] \stackrel{(\ref{eq:ddgrl1283192837921})}{\geq} \left( 1- \sfrac{1}{e} \right) \sum_{u \in B} \Pr[u \in \mc{A}(\Gamma_p)] = \left( 1- \sfrac{1}{e} \right) \E|M_\mc{A}|,
$$
completing the proof.

\smparagraph{Property~\ref{viprop:Asideprob}.} If a vertex $v \in A$ does not propose, it remains unmatched in $M_\mc{B}$. Thus:
$$
\Pr[v \in M_\mc{B}] \leq \Pr[v \text{ proposes}] = \E[p_v] = \Pr[v \in M_\mc{A}].
$$

\smparagraph{Property~\ref{viprop:Bsideprob}.} For any vertex $u\in B$,
$$
	\Pr[u \in M_\mc{B}] \leq \Pr[u \text{ receives a proposal}] \leq \sum_{v \in N(u)} \E[p_{(v, u)}] = \Pr[u \in M_\mc{A}].
$$

\smparagraph{Property~\ref{viprop:independence}.} Observe that in order to determine $u \in M_\mc{B}$, it suffices to reveal the proposals of its neighbors in $A$. If one of them proposes to $u$ then $u$ is matched and otherwise it is not. Now since $(u, v) \not \in E$ as assumed by Property~\ref{viprop:independence}, the proposals of $v$ remains completely unknown and independent of $u \in M_\mc{B}$.
\end{proof}

\subsection{Proving Lemma~\ref{lem:appvil} via the Vertex-Independent Lemma~\ref{lem:vertexindependent}}\label{sec:proofofappvil}

To prove Lemma~\ref{lem:appvil}, we prove two different bounds on the expected size of $\mu(G_p)$. The first one is easy to prove and is as follows: 

\begin{claim}\label{cl:matchingQ}
	$\E[\mu(G_p)] \geq \b{q}(Q).$
\end{claim}
\begin{proof}
	As assumed in Lemma~\ref{lem:appvil}, $q_e = \Pr[e \in \mc{M}(Q_p)]$ for any $e \in Q$. By linearity of expectation, this implies $
	\E|\mc{M}(Q_p)| = \b{q}(Q)$. Since $Q_p \subseteq E_p$, any edge in $\mc{M}(Q_p)$ appears in $G_p$ and thus the same lower bound also holds for $\E[\mu(G_p)]$ completing the proof.
\end{proof}

The second bound is the main part of the proof, and reads as follows:

\begin{claim}\label{cl:matchingQandS}
	$\E[\mu(G_p)] \geq (1-6\epsilon) \Big( \frac{e-1}{e} \cdot \b{q}(Q) + \b{q}(S) \Big).$
\end{claim}

Let us first see how the combination of Claims~\ref{cl:matchingQ} and \ref{cl:matchingQandS} proves Lemma~\ref{lem:appvil}. Observe that since each edge of $G$ is either in $Q$ or $S$, $|\b{q}| = \b{q}(Q) + \b{q}(S)$. Now consider two cases:

\smparagraph{Case 1 --- $\b{q}(Q) \geq \frac{e}{e+1} |\b{q}|$:} In this case, Claim~\ref{cl:matchingQ} already implies Lemma~\ref{lem:appvil}.

\smparagraph{Case 2 --- $\b{q}(Q) < \frac{e}{e+1} |\b{q}|$:} In this case, by Claim~\ref{cl:matchingQandS}, we have
$$
\E[\mu(G_p)] \geq (1-6\epsilon) \left(\frac{e-1}{e} \cdot \b{q}(Q) + \b{q}(S) \right).
$$
The right hand side is minimized, when as much of the weight of $\b{q}$ comes from $Q$ instead of $S$. But, by the assumption of Case 2, $\b{q}(Q) < \frac{e}{e+1} |\b{q}|$. Thus:
$$
\E_{G_p}[\mu(G_p)] \geq (1-6\epsilon) \left( \frac{e-1}{e} \cdot \frac{e}{e+1} \cdot |\b{q}| + \frac{1}{e+1}\cdot |\b{q}|\right) = (1-6\epsilon) \frac{e}{e+1} |\b{q}|,
$$
which is the desired bound of Lemma~\ref{lem:appvil}.

\subsubsection{Proof of Claim~\ref{cl:matchingQandS}}

In order to argue that $G_p$ has a matching of our desired expected size, we construct a fractional matching $\b{x}$ on it. Since the graph is bipartite, any fractional matching can be turned into an integral matching of at least the same size. As a result, it suffices to argue that fractional matching $\b{x}$ has our desired size in expectation. 

To construct fractional matching $\b{x}$, we first use the vertex-independent matching Algorithm~\ref{alg:B} of Lemma~\ref{lem:vertexindependent} to construct an integral matching $M_\mc{B}$ on $Q_p$ (the parameters that we feed into Lemma~\ref{lem:vertexindependent} are formalized below). We then use $M_\mc{B}$ to define $\b{y}: E \to \mathbb{R}_+$ which will be very close to our final fractional matching $\b{x}$, except that for a small fraction of vertices, the value of $y_v := \sum_{e \ni v} y_e$ may exceed one due to deviations in our random process. We then scale down $\b{y}$ to obtain our fractional matching $\b{x}$ and finally argue that $\b{x}$ is large enough.

\newcommand{\vprop}[0]{\ensuremath{v \textsf{ prop}}}
\newcommand{\notvprop}[0]{\ensuremath{\overline{v \textsf{ prop}}}}
\newcommand{\notxprop}[1]{\ensuremath{\overline{#1 \textsf{ prop}}}}

For brevity, we use $(\vprop)$ and $(\notvprop)$ to indicate respectively the events that a vertex $v \in A$ proposes and does not propose in Algorithm~\ref{alg:B} for constructing $M_\mc{B}$.

The formal definition of $\b{y}$, given matching $M_\mc{B}$ is given below:
\begin{equation}\label{eq:def-y}
	y_e \gets \begin{cases}
		1 & \text{if $e \in M_\mc{B}$ (this implies $e \in Q_p$)}\\
		q_e /(p \Pr[\notvprop] \Pr[u \not\in M_\mc{B}]) & \text{if $e \in S_p$, $u \not\in M_\mc{B}$, and $\notvprop$}\\
		0 & \text{otherwise,}
	\end{cases}
	\qquad \parbox{3.2cm}{
	$\forall e=(u, v) \in E$,\\
	$v \in A$, $u\in B$.
	}
\end{equation}

As discussed, $\b{y}$ is not necessarily a valid fractional matching since for some vertices $v$, $y_v$ may be larger than 1 due to some low probability (but still likely to occur) events. To resolve this, we define the final, always valid, fractional matching $\b{x}$ as follows:
\begin{equation}
	x_e \gets \begin{cases}
		y_e / (1+\epsilon) & \text{if $y_v \leq 1+\epsilon$ and $y_u \leq 1 + \epsilon$,}\\
		0 & \text{otherwise,}
	\end{cases}
	\qquad \forall e = (u, v) \in E.
\end{equation}


\smparagraph{The Matching $M_\mc{B}$:} As discussed, we use Lemma~\ref{lem:vertexindependent} to construct matching $M_\mc{B}$. In order to use this lemma, we have to specify: $(i)$ what graph we run the matching algorithm of this lemma on, and $(ii)$ what algorithm $\mc{A}$ we feed into the lemma. For the first question, the graph $\Gamma$ on which we run the lemma is simply the subgraph $Q$ of $G$, and we let $\Gamma_p$ correspond to the realized edges $Q_p$. This ensures that the reported matching $M_\mc{B}$ of Lemma~\ref{lem:vertexindependent} satisfies $M_\mc{B} \subseteq Q_p$. For the second question, we first define an auxiliary matching $M'_\mc{A} := \mc{M}(Q_p)$ and let $M_\mc{A}$ include each edge of $M'_\mc{A}$ independently with probability $1-\epsilon$.  This  down sampling step is rather technical and is there to just ensure $\Pr[w \in M_\mc{A}] \leq 1-\epsilon$ for any vertex $w$. 

Defining $M_\mc{A}$ this way, as a corollary of Lemma~\ref{lem:vertexindependent} we get:
	
\begin{corollary}\label{cor:dchuaoe18279312873}
	The matching $M_\mc{B}$ constructed as above on subgraph $Q_p$ satisfies:
	\begin{enumerate}[label=$(\roman*)$]
\item $\E|M_{\mc{B}}| \geq (1-\frac{1}{e}) \E|M_\mc{A}| = (1-\epsilon)(1-\frac{1}{e}) \cdot \b{q}(Q)$.
\item For any vertex $v \in A$, $\Pr[v \in M_\mc{B}] \leq \Pr[\vprop] = \Pr[v \in M_\mc{A}] = (1-\epsilon) \Pr[v \in \mc{M}(Q_p)]$.
\item For any vertex $u \in B$, $\Pr[u \in M_{\mc{B}}] \leq \Pr[u \in M_\mc{A}] = (1-\epsilon) \Pr[u \in \mc{M}(Q_p)]$.
\item For any $v \in A, u \in B$ with $(u, v) \in S$, events $(\vprop)$ and $(u \in M_\mc{B})$ are independent.
\end{enumerate}
\end{corollary}
\begin{proof}
	The first three are simply followed by the properties of Lemma~\ref{lem:vertexindependent} combined with the definition of $M_\mc{A}$ above. The last property holds since $(u, v) \in S$ implies $(u, v) \not\in Q$ (since each there are no parallel edges and each edge belongs to exactly one of $Q$ and $S$) which combined with Property~\ref{viprop:independence} of Lemma~\ref{lem:vertexindependent} implies the stated independence.
\end{proof}

Having defined $M_\mc{B}$, the definitions of $\b{y}$ and $\b{x}$ are complete. We now turn to analyze their size. We first analyze the size of $\b{y}$ and then prove the size of $\b{x}$ is actually very close to that of $\b{y}$ by bounding the probability of deviations leading to vertices to have $y_v > 1+\epsilon$.

\smparagraph{The Expected Size of $\b{y}$:} We have $\E|\b{y}| = \sum_{e \in E} \E[y_e] = \sum_{e \in Q} \E[y_e] + \sum_{e \in S} \E[y_e]$. For edges $e \in Q$ we have $y_e = 1$ iff $e \in M_\mc{B}$, which combined with $M_\mc{B} \subseteq Q_p$ implies $\sum_{e \in Q} \E[y_e] = \E|M_\mc{B}|$. On the other hand, by definition of $\b{y}$ on edges in $S$, we have
\begin{flalign}
	\nonumber \E|\b{y}| &= \E|M_\mc{B}| + \sum_{e \in S} \Pr[e \in S_p, u \not\in M_\mc{B}, \notvprop] \cdot \frac{q_e}{p \Pr[u \not\in M_\mc{B}] \Pr[\notvprop]}.
\end{flalign}
For any $e \in S$, $e \in S_p$ iff $e$ is realized which is independent of how matching $M_\mc{B}$ is constructed. This holds because in constructing matching $M_\mc{B}$ we do not look at the realized edges $S_p$ of $S$. On the other hand, by Corrolary~\ref{cor:dchuaoe18279312873} Property $(iv)$, $u \in M_\mc{B}$ and $\notvprop$ are also independent, hence
\begin{flalign}
	\nonumber \E|\b{y}| &= \E|M_\mc{B}| + \sum_{e \in S} q_e\\
	\nonumber &\geq (1-\epsilon)(1-\sfrac{1}{e}) \cdot \b{q}(Q) + \sum_{e \in S} q_e \tag{By Corollary~\ref{cor:dchuaoe18279312873} Property $(i)$}\\
	&\geq (1-\epsilon) \Big( (1-\sfrac{1}{e}) \cdot \b{q}(Q) + \b{q}(S) \Big).\label{eq:ggggf1cgf192837}
\end{flalign}

Note that $\b{y}$, as proved above, is as large as the lower bound required by Claim~\ref{cl:matchingQandS}. However, we want to show that $\E|\b{x}|$ is also this large (up to $1-\Theta(\epsilon)$ factors). This is what we prove next.

\smparagraph{The Expected Size of $\b{x}$:} Take an edge $e=(v, u) \in E$. Observe from the construction of $\b{x}$ that either $x_e = y_e / (1+\epsilon)$, or $x_e = 0$ if $y_v > 1+\epsilon$ or $y_u > 1 + \epsilon$. This implies that
$$
	\E|\b{x}| = \sum_{e \in E} \E[x_e] = \sum_{e=(u, v) \in E} \E[y_e/(1+\epsilon) \mid y_v \leq 1+\epsilon, y_u \leq 1+\epsilon].
$$
Now an edge $e \in E$ either belongs to $Q$ or $S$. For edges $e = (u, v) \in Q$, it is not hard to see that we have $x_e = y_e / (1+\epsilon)$ with probability one, and thus 
\begin{equation}\label{eq:rrll198723}
\E[x_e] = \E[y_e]/(1+\epsilon)	 \qquad \forall e \in Q.
\end{equation}
The reason behind this, is that if $e =(u, v) \in Q$ then by construction of $\b{y}$, $y_e = \pmb{1}(e \in M_\mc{B})$. Further, if it occurs that $y_e = 1$, then both of its endpoints are matched in $M_\mc{B}$ and hence we have $y_{e'} = 0$ for all other edges $e'$ connected to either $u$ or $v$, again by construction of $\b{y}$. As a result, $y_e = 1$ implies $y_u = y_v = 1$ and thus $x_e = y_e / (1+\epsilon)$; otherwise $x_e = y_e = 0$, proving (\ref{eq:rrll198723}).

The situation, however, is more complicated for edges $e \in S$ as in this case $\b{y}$ on the endpoints of $e$ may exceed one. We show, however, that this occurs with a small enough probability that we can still argue that it does not affect the expected value of $x_e$ by much, and $\E[x_e] \approx \E[y_e]$.

\begin{claim}\label{cl:u-frac-exp}
	For any vertex $u \in B$, $\E[y_u \mid u \not\in M_\mc{B}] \leq q_u$.
\end{claim}
\begin{proof}
	Let $v_1, \ldots, v_d$ be the neighbors of $u$ in $S$, and let $e_i = (u, v_i)$. We have
	$$
		\E[y_u \mid u \not\in M_\mc{B}] = \sum_{i=1}^d \frac{q_{e_i}}{p \Pr[u \not\in M_\mc{B}] \Pr[\notxprop{v_i}]} \cdot \Pr[e_i \in S_p, u \not\in M_\mc{B}, \notxprop{v_i} \mid u \not\in M_\mc{B}].
	$$
	By independence of events in the probability (justified before), we can simplify this to
	$$
	\E[y_u \mid u \not\in M_\mc{B}] = \sum_{i=1}^d \frac{q_{e_i}}{\Pr[u \not\in M_\mc{B}]} = \frac{\sum_{i=1}^d q_{e_i}}{\Pr[u \not\in M_\mc{B}]}.
	$$
	Now let $q^S_u$ denote the sum of fractional value written on edges of $u$ in $S$ and let $q^Q_u$ denote the same but on edges of $v$ in $Q$. The nominator equals $q^S_u$. Also from Corollary~\ref{cor:dchuaoe18279312873} Property $(iii)$, thus $\Pr[u \in M_\mc{B}] \leq (1-\epsilon) q^Q_u$, we get $\Pr[u \not\in M_\mc{B}] \geq 1 - (1-\epsilon) q^Q_u \geq 1 - q^Q_u$. Combining them gives
	$$
		\E[y_u \mid u \not\in M_\mc{B}] \leq \frac{q^S_u}{1-q^Q_u} = \frac{q_u \cdot q^S_u}{q_u (1-q_u^Q)} = \frac{q_u \cdot q^S_u}{q_u-q_u \cdot q_u^Q}  \leq \frac{q_u \cdot q^S_u}{q_u - q^Q_u} = \frac{q_u \cdot q^S_u}{q^S_u} = q_u,
	$$
	as desired.
\end{proof}

\begin{claim}\label{cl:u-frac-concentration}
	For any vertex $u \in B$, $\E[y_u \mid y_u \leq 1+\epsilon, u \not\in M_\mc{B}] \geq \E[y_u] - \epsilon q_u.$
\end{claim}
\begin{proof}
	As proved in Claim~\ref{cl:u-frac-exp}, $\E[y_u \mid u \not\in M_\mc{B}] \leq q_u$. We prove the desired inequality of the claim via a concentration bound on random variable $y'_u := (y_u \mid u \not\in M_\mc{B})$.
	
	Let $v_1, \ldots, v_d$ be the neighbors of $u$ in $S$, and let $e_i = (u, v_i)$. Let us for simplicity define random variable $y'_{e_i} := (y_{e_i} \mid u \not\in M_\mc{B})$. Since we have conditioned on $u \not\in M_\mc{B}$, we get by definition of $\b{y}$ that
	$
		y'_u = \sum_{i=1}^d y'_{e_i}.
	$ Now we argue that $y'_{e_1}, \ldots, y'_{e_d}$ are actually independent. To see this, observe that once we condition on $u \not\in \mc{B}$, the only random process that determines each $y'_{e_i}$ is whether edge $e_i$ is realized (i.e. $e_i \in S_p$) and if vertex $v_i$ proposes. Both events are independent of $u \not\in M_\mc{B}$ since, recall, to determine $u \not\in M_\mc{B}$ we only need to know the proposals of its neighbors in $Q$, and $v_1, \ldots, v_d$ are all neighbors of $u$ in $S$. 
	
	This independence allows us to prove a concentration bound on $y'_u$ and prove the claim. We do this via the second moment method. We have
	$$
		\Var[y'_u] = \sum_{i=1}^d \Var[y'_{e_i}] = \sum_{i=1}^d \E[(y'_{e_i})^2] - \E[y'_{e_i}]^2 \leq \sum_{i=1}^d \E[(y'_{e_i})^2]  \stackrel{\text{see below}}{\leq} \tau \cdot \E[y'_u]  \stackrel{\text{Claim~\ref{cl:u-frac-exp}}}{\leq} \tau \cdot q_u,
	$$
	where $\tau$ is the maximum possible outcome of $y'_{e_i}$ for any $i \in [d]$. 
	
	Plugging this into Chebyshev's inequality, we get
	$$
		\Pr[y'_u > \E[y'_u]+\delta] \leq \frac{\Var[y'_u]}{\epsilon^2} \leq \frac{\tau q_u}{\delta^2}.
	$$
	Finally, for each edge $e_i$, by construction of $\b{y}$, $y_{e_i} \leq q_{e_i}/(p\Pr[\notxprop{v_i}]\Pr[u \not\in M_\mc{B}]) \leq q_{e_i}/p\epsilon^2$ where the latter follows from Corollary~\ref{cor:dchuaoe18279312873}. Combined with $q_{e_i} \leq \epsilon^5 p$ by Lemma~\ref{lem:appvil} and $e_i \in S$, we get $\tau \leq \epsilon^3$. We thus get $\Pr[y'_u > 1+\epsilon] \leq \Pr[y'_u > \E[y'_u]+\epsilon] \leq \frac{\epsilon^3}{\epsilon^2} q_u$. This concentration of $y'_u$ implies $\E[y_u \mid y_u>1+\epsilon, u \not\in M_\mc{B}] \leq \epsilon q_u$ and thus the stated bound of the claim.
\end{proof}

Similarly, for each vertex $v \in A$ we can first bound the expected value of $y_v$ conditioned on $(\notvprop)$ and then prove this is concentrated, and overall argue $\b{x}$ remains close to $\b{y}$. Namely:

\begin{claim}\label{cl:yv-concentrate}
	For any $v \in A$, $\E[y_v \mid y_v \leq 1+\epsilon,\notvprop] \geq \E[y_v] - \epsilon q_v$.
\end{claim}

Since the proof is similar to Claim~\ref{cl:u-frac-concentration} for vertices in $B$ and there are only minor differences, we defer it to Appendix~\ref{sec:proofs}.

To complete the proof that $\E|\b{x}| \approx \E|\b{y}|$, note that
\begin{flalign*}
	\E|\b{x}| &= \frac{1}{2} \sum_{w \in A \cup B} \E[x_w] =  \frac{1}{2} \sum_{w \in A \cup B} \frac{\E[y_w \mid y_w \leq 1+\epsilon]}{1+\epsilon} \stackrel{\text{Claims~\ref{cl:u-frac-concentration}, \ref{cl:yv-concentrate}}}{\geq} \left(\frac{1}{2} \sum_{w \in A \cup B} \frac{\E[y_w]}{1+\epsilon} \right) - 3\epsilon |\b{q}|\\
	&\geq \frac{1}{1+\epsilon} |\b{y}| - 3\epsilon |\b{q}| \stackrel{(\ref{eq:ggggf1cgf192837})}{\geq} (1-\epsilon)\left( (1-\epsilon) \Big( (1-\sfrac{1}{e}) \cdot \b{q}(Q) + \b{q}(S) \Big) \right) - 3\epsilon |\b{q}|\\
	&\geq (1-6\epsilon)\Big( (1-\sfrac{1}{e}) \cdot \b{q}(Q) + \b{q}(S) \Big).
\end{flalign*}
Since as discussed $\E[\mu(G_p)] \geq \E|\b{x}|$, this completes the proof of Claim~\ref{cl:matchingQandS}.

\section{Upper Bounds}\label{sec:upperbounds}
In this section, we present our main algorithms for the stochastic vertex cover problem.
\subsection{Bipartite Graphs: $1.36$-Approximation with $\poly(\frac{1}{p})$ Queries}

Our main result in this section is the following stochastic vertex cover result for bipartite graphs:

\begin{theorem}\label{thm:bipartitevc}
	For any $p \in (0, 1]$, any bipartite graph $G=(V, E)$ has a subgraph $Q$ of maximum degree $O(1/p^6)$ where querying only the edges in $Q$ suffices to find $C \subseteq V$ such that:
	\begin{enumerate}[itemsep=0pt, topsep=5pt]
		\item $C$ is a vertex cover of $G_p$ with probability 1.
		\item The expected size of $C$ is at most 1.367 $(\approx \frac{e+1}{e})$  times the size of $\E[\nu(G_p)]$.
		\item Both $Q$ and $C$ can be found in polynomial time.
	\end{enumerate}
\end{theorem}

To prove Theorem~\ref{thm:bipartitevc}, we first use the  Half-Stochastic Matching Lemma~\ref{lem:vcpartition} to obtain a partitioning $(Q, S)$ of $E$. We then query the edges in $Q$, and report the MVC of $H = Q_p \cup S$ as the vertex cover for $G_p$. We finally use the algorithm $\mc{M}$ provided by this lemma, as well as Lemma~\ref{lem:appvil} to analyze the approximation ratio of this algorithm.

The formal algorithm is as follows:

\begin{tboxalg2e}{The algorithm for Theorem~\ref{thm:bipartitevc}.}
\begin{algorithm}[H]
	\DontPrintSemicolon
	\SetAlgoSkip{bigskip}
	\SetAlgoInsideSkip{}
	
	\label{alg:bipartite}

	Let $Q, S$ be the partitioning found by Lemma~\ref{lem:vcpartition} for the following parameters: $G$ is the given base graph, $p$ is the realization probability, and let $\epsilon > 0$ be a sufficiently small constant which adjusts how close the approximation will be to $\frac{e+1}{e}$.
	
	Query the edges in $Q$ and let $Q_p$ be the edges in $Q$ that are realized.
	
	Return a minimum vertex cover $C$ of graph $H = Q_p \cup S$.
\end{algorithm}
\end{tboxalg2e}

Intuitively, what we do in Algorithm~\ref{alg:bipartite} is to query only the edges in $Q$, assume that the rest of the edges in $S$ are all realized, and report a MVC of the resulting graph $H = Q_p \cup S$. Note that the vast majority of the edges in $G$ belong to $S$ since $Q$ has only a constant maximum degree. As a result, we have to argue that the extra constraints imposed by assuming that all these edges are realized, do not increase the vertex cover size by much. 

Before analyzing the size of the vertex cover $C$, let us explain why it is always a valid vertex cover of $G_p$ and analyze the query-complexity of Algorithm~\ref{alg:bipartite}. These will be simple consequences of the properties provided by Lemma~\ref{lem:vcpartition}.

\smparagraph{Query-Complexity and Validity of the Vertex Cover:} By Property~\ref{vp-prop:degreeQ} of Lemma~\ref{lem:vcpartition} the maximum degree in $Q$ is at most $O(1/\epsilon^{11}p^6)$; thus, we query at most $O(1/p^6)$ edges per vertex. The validity of the vertex cover is also easy to confirm. By definition of $H$, an edge $e \in E$ is not in $H$ if and only if it belongs to $Q$ and is not realized. Therefore, all realized edges must belong to $H$ implying that $G_p$ is a subgraph of $H$. As a result, a vertex cover of $H$ covers all the edges in $G_p$.

\smparagraph{The Approximation Ratio:} We use Lemma~\ref{lem:appvil}, as well as the properties of the partitioning provided by Lemma~\ref{lem:vcpartition}, to show that the expected size of the MVC in graph $H$, which is reported by Algorithm~\ref{alg:bipartite} as the output, is not larger than (almost) $\frac{e+1}{e}$ times the expected size of the MVC in the actual realization $G_p$, namely that for any arbitrarily small constant $\delta'$ (affecting $\epsilon$ in Algorithm~\ref{alg:bipartite}) it holds that:
$$
	\E_H[\nu(H)] \leq (1+\delta') \frac{e+1}{e}  \cdot \E_{G_p}[\nu(G_p)].
$$
Since the graph is bipartite, by K\"{o}nig's theorem, the size of maximum matching and MVC are the same. As such, it suffices to prove
\begin{equation}\label{eq:hehuc12839123}
	\E_H[\mu(H)] \leq (1+\delta') \frac{e+1}{e}  \cdot \E_{G_p}[\mu(G_p)].
\end{equation}
In order to show this, we define a fractional matching $\b{q}$ on $G$ with size $|\b{q}| \geq (1-2\epsilon) \E[\mu(H)]$. We then show that this fractional matching $\b{q}$ satisfies the required properties of Lemma~\ref{lem:appvil} and as a result, implies $\E[\mu(G_p)]$ has the desired size of (\ref{eq:hehuc12839123}).

\smparagraph{The Fractional Matching $\b{q}$:} Consider (random) matching $\mc{M}(H)$ and recall that $\mc{M}$ is the matching algorithm provided by Lemma~\ref{lem:vcpartition}. For each edge $e \in E$, we simply let $q_e \gets \Pr[e \in \mc{M}(H)]$. Clearly $\b{q}$ is a valid fractional matching since the probabilities around each vertex correspond to its probability of being matched and thus do not exceed one. Moreover, $|\b{q}| = \E[\mc{M}(H)]$ by linearity of expectation, combined with  Lemma~\ref{lem:vcpartition} Property~\ref{vp-prop:expmatching} that $\E|\mc{M}(H)| \geq (1-2\epsilon)\E[\mu(H)]$, we get $|\b{q}| \geq (1-2\epsilon)\E[\mu(H)]$; hence $|\b{q}|$ has the claimed size too. It remains to prove the two assumptions of Lemma~\ref{lem:appvil} are also satisfied by $\b{q}$. The first one requires us to give a matching algorithm $\mc{M}'$ where $q_e = \Pr[e \in \mc{M}'(Q_p)]$ for all $e \in Q$. Letting $\mc{M}'(Q_p) := \mc{M}(H) \cap Q_p$ suffices for this purpose, since recall that any edge in $\mc{M}(H) \cap Q$ is already in $Q_p$ by definition of graph $H$ in Lemma~\ref{lem:vcpartition} and thus $\Pr[e \in \mc{M}'(Q_p)] = q_e$. For the second assumption we need $q_e \leq \delta^5 p$ for all $e \in S$ (here we used $\delta$ instead of $\epsilon$ to avoid confusion with parameter $\epsilon$ that we feed into Lemma~\ref{lem:vcpartition}). Indeed, by Property~\ref{vp-prop:prsmallS}, we have $q_e \leq \epsilon^2 p$ and it suffices to let $\delta = \epsilon^{2/5}$. We can, now, apply Lemma~\ref{lem:appvil} to obtain:
$$
	\E[\mu(G_p)] \geq (1-4\delta) \frac{e}{e+1}|\b{q}| \geq (1-4\delta)(1-2\epsilon) \frac{e}{e+1} \E[\mu(H)].
$$
Since we can let $\epsilon$ (and thus $\delta$) be any desirably small constant, this proves (\ref{eq:hehuc12839123}) and our claim that our reported vertex cover has size at most (arbitrarily close to) $\frac{e+1}{e} \E[\nu(G_p)]$.

\smparagraph{Stochastic Matchings:} Finally, we note that the same tools we used for this problem also lead to Corollary~\ref{cor:matching}. To prove it, we in fact, provide a novel analysis for a well-known Monte Carlo algorithm for the stochastic matching problem that is very different from the algorithm we use for Theorem~\ref{thm:bipartitevc}. This is why the number of queries in Corrolary~\ref{cor:matching} and Theorem~\ref{thm:bipartitevc} are different. But the new analysis, is also based on Lemma~\ref{lem:appvil} similar to above. 

We analyze the following algorithm proposed first in \cite{soda19} and further analyzed in \cite{stoc20,focs20}.

\begin{tboxalg2e}{A Monte-Carlo algorithm for stochastic matching \cite{soda19}.}
\begin{algorithm}[H]
	\DontPrintSemicolon
	\SetAlgoSkip{bigskip}
	\SetAlgoInsideSkip{}
	
	\label{alg:montecarlo}

	For large enough $R = O(\frac{\log 1/p}{p})$, take independent realizations $G_1, \ldots, G_R$ of $G$.
	
	For some deterministic maximum matching algorithm $\MM{\cdot}$, query subgraph $Q := \MM{G_1} \cup \ldots \cup \MM{G_R}$ of $G$ and report the maximum matching in $Q_p$.
\end{algorithm}
\end{tboxalg2e}

The following lemma is implied in \cite{soda19} for Algorithm~\ref{alg:montecarlo}:
\begin{claim}[\cite{soda19}]\label{cl:stochasticmatchingfractional}
	For any desirably small constant $\epsilon > 0$ (affecting the hidden constants in $R$), there is a partitioning $C, N$ of $Q$, and a fractional matching $\b{q}$ on $Q$ such that $(i)$ $\E|\b{q}| \geq (1-\epsilon) \E[\mu(G_p)]$, $(ii)$ for each edge $e \in N$, $q_e \leq \epsilon^5 p$, $(iii)$ for each $e \in C$, $q_e = \Pr[e \in \MM{G_p}]$.
\end{claim}

Here we briefly describe the intuition behind Claim~\ref{cl:stochasticmatchingfractional}. For the complete proof see \cite{soda19}.

\begin{proof}[Proof sketch of Claim~\ref{cl:stochasticmatchingfractional}]
Partition the edges of graph $G$ into two subsets: {\em crucial} edges $e$ defined as those with $\Pr[e \in \MM{G_p}] \geq \tau$ for a small enough parameter $\tau = \epsilon^{O(1)} p$, and {\em non-crucial} edges which include all the rest of edges $e$ with $\Pr[e \in \MM{G_p}] < \tau$. 

Let subgraphs $C$ and $N$ of Claim~\ref{cl:stochasticmatchingfractional} be respectively the set of crucial and non-crucial edges of $G$ that belong to $Q$. We define $\b{q}$ such that the size of $\b{q}$ on $C$ is $(1-\epsilon)$ times the expected contribution of crucial edges to $\MM{G_p}$ and, similarly, the size of $\b{q}$ on $N$ is $(1-\epsilon)$ times the expected contribution of non-crucial edges to $\MM{G_p}$. This way, we guarantee property $(i)$.

To define $\b{q}$ on $C$, for any edge $e \in C$ we let $q_e = \Pr[e \in \MM{G_p}]$. It can be easily confirmed that if the parameter $R$ of Algorithm~\ref{alg:montecarlo} is larger than $\log(1/\epsilon)/\tau$, then each crucial edge is added to $Q$ with probability at least $1-\epsilon$. As such, fractional matching $\b{q}$ on $C$ has size at least $(1-\epsilon)$ fraction of the expected contribution of crucial edges to $\MM{G_p}$. 
 
 On the flip side, however, $Q$ includes only a small fraction of non-crucial edges. Hence, to maintain the property that $\b{q}$ on $N$ is almost as large as the expected contribution of non-crucial edges to $\MM{G_p}$, the value of $q_e$ must be much larger than $\Pr[e \in \MM{G_p}]$ for $e \in N$. To do this, suppose that we define $q_e$ to be the {\em fraction} of matchings $\MM{G_1}, \ldots, \MM{G_R}$ that include $e$. Observe that each edge $e$ in the graph, crucial or non-crucial, is expected to appear in exactly $\Pr[e \in \MM{G_p}]$ fraction of matchings $\MM{G_1}, \ldots, \MM{G_R}$ since each one includes $e$ with probability exactly $\Pr[e \in \MM{G_p}]$. Hence, by defining $\b{q}$ on $N$ this way, we get that the expected size of $\b{q}$ on $N$ is at least the expected contribution of $N$ to $\MM{G_p}$. Unfortunately, however, $\b{q}$ may violate fractional matching constraints with this construction. Namely, that $q_v$ for a vertex $v$ may exceed one. The next important observation is that this violation cannot be too large, since the fraction of matchings $\MM{G_1}, \ldots, \MM{G_R}$ in which a vertex $v$ is matched via a non-crucial edge is sufficiently concentrated around the probability that $v$ is matched via a non-crucial edge in $\MM{G_p}$. The final fractional matching is obtained by slightly modifying this fractional matching (particularly by discarding vertices that deviate too much and multiplying the rest of the values by some $(1-\epsilon)$ factor) so that no constraints are violated.
\end{proof}

Having it, we can now plug $\b{q}$ into Lemma~\ref{lem:appvil} and obtain that $\E[\mu(Q_p)] \geq (1-5\epsilon) (\frac{e}{e+1})\E[\mu(G_p)]$, thereby proving Corollary~\ref{cor:matching}.

\newcommand{\WF}[1]{\ensuremath{\textsf{Filling}(#1)}}

\subsection{General Graphs: $(2+\epsilon)$-Approximation with $O(\frac{1}{p})$ Per-Vertex Queries}

In this section, we prove the following result:
\begin{theorem}\label{thm:generalvc}
	For any $\epsilon > 0$ and $p \in (0, 1]$, any (general) graph $G=(V, E)$ has a subgraph $Q$ of maximum degree $O(\frac{1}{\epsilon^3 p})$ where querying only the edges in $Q$ suffices to find $C \subseteq V$ such that:
	\begin{enumerate}[itemsep=0pt, topsep=5pt]
		\item $C$ is a vertex cover of $G_p$ with probability 1.
		\item The size of $C$ is in expectation at most $(2+\epsilon)$  times the minimum vertex cover of $G_p$.
		\item It is possible to find $Q$ and $C$ in polynomial time.
	\end{enumerate}
\end{theorem}

We start with a subroutine for constructing a fractional matching on a given graph, and then describe our algorithm which proves Theorem~\ref{thm:generalvc}.

\smparagraph{A fractional matching subroutine:} Consider a simple and well-known fractional matching algorithm which starts with a zero-size fractional matching and gradually increases the fractional values on the edges all at the same (additive) rate. Once the fractional value around a vertex reaches one, we mark this vertex as {\em inactive} and stop increasing the fractional value of its edges. We will use a  slightly different variant of this algorithm in Algorithm~\ref{alg:filling} where the vertices may be made inactive sooner; i.e., once they reach a given budget; this variant is formalized below. 

Let $G$ be a graph, and for each vertex $v$, let $b(v) \in (0, 1]$ be a given {\em budget}. Initially, every vertex $v$ is {\em active} and for each edge $e$ we set $x_e := 0$. The algorithm proceeds in at most $n$ steps. In each step, for any edge $e$ whose both endpoints are active, we increase $x_e$ for all the edges in the same rate until for a vertex we have $x_v := \sum_{e \ni v} x_e= b(v)$. When this event happens for a vertex $v$ it  becomes  inactive, which implies that for its edges $e$, $x_e$ will no longer change. Algorithm~\ref{alg:filling} is the pseudo-code of this process.

It is clear that throughout the algorithm we have $x_v \leq b(v) \leq 1$ for every vertex $v$. Hence, at every point in the algorithm, vector $\b{x} := (x_e)_{e \in E}$ is a fractional matching of $G$.  Let us define $\b{x}^{(t)}$ to be equal to $\b{x}$ from an iteration of the algorithm after which for at least one edge we have $x_e > t$. (If $x_e > t$ never happens then $\b{x}^{(t)}$, is from the last iteration of the algorithm.)  We can intuitively think of the algorithm above as a continuous process over a time interval of $[0, 1]$ that gradually increases the fractional matching on all edges with active endpoints, all at the same additive rate, until every vertex becomes inactive. The value of $\b{x}^{(t)}$ can thus be interpreted as the fractional matching constructed by time $t$ of this process. 

\begin{tboxalg2e}{\WF{G=(V, E), b: V \to (0, 1]}}
\begin{algorithm}[H]
	\DontPrintSemicolon
	\SetAlgoSkip{bigskip}
	\SetAlgoInsideSkip{}
	
	\label{alg:filling}
	For any edge $e \in E$, set $x_e := 0$.\;

	\Repeat{All the vertices are inactive}{
	Call a vertex  {\em inactive} iff  $x_v := \sum_{e \ni v} x_e = b(v)$, and {\em active} otherwise.\;
	Call an edge  active iff both its endpoints are active, and inactive otherwise. \;
	 Pick the minimum parameter $\delta\in (0,1)$ such that setting $x_e \gets x_e+\delta$ for all the active edges, results in at least one new inactive vertex.
		\;
		Set $x_e \gets x_e+\delta$ for all active edges.\;
	}
	
	\Return fractional matching $\b{x} \gets (x_e)_{e \in E}$.
\end{algorithm}
\end{tboxalg2e}

\newcommand{\tthreshold}[0]{\ensuremath{\epsilon^3 p}}

\smparagraph{Our stochastic vertex cover algorithm.} We use Algorithm~\ref{alg:general} to decide which edges to query and which vertices to put in the vertex cover. Here we give an informal overview of this algorithm.

The algorithm starts by running Algorithm~\ref{alg:filling} on the static graph $G$ with a budget of 1 per vertex, to obtain a fractional matching $\b{x}$. We will, in fact, only need the fractional matching $\b{x}'$ constructed by this algorithm up to time $t := \Theta(\tthreshold)$, i.e. $\b{x}' := \b{x}^{(t)}$. Once we have $\b{x}'$ and before we query any edge, we commit the vertices with fractional matching value $1$ in $\b{x}'$ to be in our final vertex cover. Let $F := \{v \in V : x'_v = 1 \}$ denote the set of these vertices and let $Q$ be the edges in $E$ that do not have an endpoint in $F$. Note that any edge in $E \setminus Q$ is already covered by $F$; hence we do not need to query them. We, thus, only query the edges in $Q$. Let $Q_p$ be the subset of edges in $Q$ that turn out to be realized. At least one of the endpoints of each edge in $Q_p$ should be added to the vertex cover. To decide which ones join the vertex cover, we again run Algorithm~\ref{alg:filling}, but this time on subgraph $Q_p$ and we set the budget of each vertex $v$ to be $1-x'_v$. Let $\b{y}$ be the resulting fractional matching on $Q_p$. We report the set $C := \{ v \in V : x'_v + y_v = 1 \}$ as the vertex cover. (Note from the definition that $F \subseteq C$, hence satisfying our earlier claim that we ``commit'' $F$ to be in the final vertex cover.)

\begin{tboxalg2e}{The algorithm for Theorem~\ref{thm:generalvc}.}
\begin{algorithm}[H]
	\DontPrintSemicolon
	\SetAlgoSkip{bigskip}
	\SetAlgoInsideSkip{}
	
	\label{alg:general}

	Let $\b{x} := \WF{G, b}$ where $b(v) = 1$ for each $v \in V$.\;
	
	Fix $t = \Theta(\tthreshold)$ and let $\b{x}' := \b{x}^{(t)}$.
	
	Let $F := \{v \in V: x'_v = 1 \}$.\;
	
	Query edges $Q := \{e=(u, v) \in E : u \not\in F \text{ and } v \not\in F\}$ with no endpoint in $F$.\;
	
	Let $Q_p$ be the realized edges in $Q$.\;
	
	Run $\b{y} := \WF{Q_p, b'}$ where $b'(v) = 1 - x'_v$ for each $v \in V$. \label{line:WF2} \;
	
	Report $C := \{ v \in V : x'_v + y_v = 1 \}$ as the vertex cover of $G_p$.
\end{algorithm}
\end{tboxalg2e}

\smparagraph{Validity of the vertex cover.} The proof of why the set $C$ reported by Algorithm~\ref{alg:general} is always a valid vertex cover of $G_p$ is simple. We start by formalizing our earlier claim that all vertices in $F$ also appear in $C$. To see this, observe that any vertex $v$ in $F$ by definition has $x'_v = 1$ which also implies $y_v = 0$ since $b'(v) = 1 - x'_v = 0$; this in turn implies $x'_v + y_v = 1$ and thus by definition $v \in C$. Now take an edge $e$ of the realization $G_p$. If $e$ has an endpoint in $F$, then this endpoint is in $C$, covering $e$. Therefore, let us take an edge $e = (u, v)$ in $G_p$ whose both endpoints are in $V \setminus F$. Observe that by definition we have $e \in Q_p$. Now once we run $\WF{Q_p, b'}$, we increase the fractional value on $e$ until one of its endpoints $v$ reaches its budget $b'(v)$. Since we have set $b'(v) = 1-x'_v$ and $y_v = b'(v)$, we have $x'_v + y_v = 1 - b'(v) + b'(v) = 1$ and thus $v$ belongs to $C$, covering $e$.

\smparagraph{Analysis of the number of queries.} Observe that in Algorithm~\ref{alg:general}, we only query the edges in $Q$. Thus, it suffices to prove that the maximum degree in $Q$ is at most $O(1/\tthreshold)$ to prove that Algorithm~\ref{alg:general} queries at most $O(1/\tthreshold)$ edges per vertex. Consider an iteration $i$ of the algorithm when $\b x= \b x^{t}$, and let $\b x'$ be the fractional matching from the next iteration (if any). By definition of $\b x^{t}$ and $Q$, at iteration $i$,  any edge $e\in Q$ has two active endpoints and thus is active itself. Moreover, by Algorithm~\ref{alg:filling}, at any iteration, all the edges that have been active in the previous iteration have the same fractional value; thus, we have $x'_e > t$ for any $e\in Q$. Moreover, as a result of $\b x'$ being a valid fractional matching, any vertex has at most $1/t = \Theta(1/\tthreshold)$ edges in $Q$.

\smparagraph{Running time of the Algorithm.} The algorithm is clearly polynomial time as in each iteration of $\WF{G, b}$, at least one vertex becomes inactive and this can happen at most $n$ times.

\smparagraph{Analysis of the approximation ratio.} The more challenging part is to prove that this vertex cover $C$ reported by Algorithm~\ref{alg:general} is in expectation at most $(2+\epsilon)$ times larger than the minimum vertex cover of $G_p$. To prove this, in Lemma~\ref{lemma:generalmatching}, we show that it is possible to construct a fractional matching $\b y$ of graph $G_p$ such that 
$(2+\epsilon)\E[|\b y|] \geq \E[|C|]$.
By weak duality, the minimum vertex cover of a graph is larger than any of its fractional matchings; hence, $C$ is a $(2+\epsilon)$-approximate minimum vertex cover. 
 \newcommand{\valuealpha}{0.25\epsilon}
 \begin{lemma} \label{lemma:generalmatching}
There exists a fractional matching $\b y$ of $G_p$ such that $(2+\epsilon)\E[|\b y|] \geq \E[|C|].$	
 \end{lemma}
 
 \begin{proof}
 We start by giving Algorithm~\ref{alg:lemmamatching}  that constructs a fractional matching of $G_p$.

\begin{tboxalg2e}{The algorithm for Lemma~\ref{lemma:generalmatching}.}
\begin{algorithm}[H]
	\DontPrintSemicolon
	\SetAlgoSkip{bigskip}
	\SetAlgoInsideSkip{}
	
	\label{alg:lemmamatching}
		For any edge $e\in S_p$ we set $y_e \gets \frac{x^{(t)}_e}{(1+\alpha) p}$ for $\alpha = \valuealpha$.\;
		
	If a vertex $v$ does not satisfy $\sum_{e \ni v}y_e\geq x^{(t)}_v$ then, for any edge $e \ni v$ set $y_e \gets 0$.\;

	For any $e\in Q_p$ set $y_e \gets x'_e$ where $\b x' \gets \WF{Q_p, b}$ with $b(v) = 1-x^{(t)}_v$ for any $v \in V\setminus F $.\;
	Report $\b y$ as the fractional matching of $G_p$.\;
\end{algorithm}
\end{tboxalg2e}
This algorithm consists of two stages. In the first stage, we construct a fractional matching on edges in $S_p$ and in the second stage, we add the edges of $Q_p$ to this matching. The first stage starts by setting 	$$x_e := \frac{x^{(t)}_e}{(1+\alpha) p}$$ for any edge $e\in S$ that is realized, where $\alpha = \valuealpha$. Let us call a vertex {\em bad} iff it satisfies $x_v:= \sum_{e \ni v}x_e>x^{(t)}_v$.
 For any bad vertex $v$ we decrease $x_e$ of all the edges $e\ni v$ to zero.

In the second stage of the algorithm, we construct a fractional matching on edges in $Q_p$. Consider $x' := \WF{Q_p, b}$ with $b(v) = 1-x^{(t)}_v$ for any $v \in Q$. Clearly, combining $\b x'$ with $\b y$ gives us a valid fractional matching since for any vertex $v\in V$ we have $\sum_{v\in V} x'_v + y_v < 1$. Hence, for any edge in $Q_p$ we set $y_e:= x'_e$. To complete the proof of this lemma, it suffices to take vertex cover $C$ outputted by Algorithm~\ref{alg:general} and  prove 
\begin{align}\label{eq:lfhguieryfg}(1+0.5\epsilon) \E\left[\sum_{v\in C} y_v\right] \geq |C|,\end{align} as it gives us 
$$ (2+\epsilon)\E[|X|] = 2(1+0.5\epsilon)\E[|X|]\geq (1+0.5\epsilon) \E\left[\sum_{v\in V} y_v\right] \geq (1+0.5\epsilon) \E\left[\sum_{v\in C} y_v\right]  \geq |C|.$$
 For any vertex $v\in C$, let $z_v := \sum_{e\ni, v, e\in S}$ and let  $w_v := \sum_{e\ni, v, e\in Q}$.
By Algorithm~\ref{alg:general}, any vertex $v\in C$ satisfies $w_v + x^{(t)}_v =1.$ In the rest of the proof we focus on proving 
\begin{align}\label{eq:perkfrj} \sum_{v\in C} x^{(t)}_v \leq \sum_{v\in C} (1+0.5\epsilon) \E[z_v],\end{align}
since it results in (\ref{eq:lfhguieryfg}) as follows: 
$$|C| = \sum_{v\in C} (w_v +  x^{(t)}_v) \stackrel{(\ref{eq:perkfrj})}{\leq}  \sum_{v\in C} w_v + (1+0.5\epsilon) \E[z_v]\leq \sum_{v\in C}(1+0.5\epsilon) \E[y_v] =(1+0.5\epsilon) \E\left[\sum_{v\in C} y_v\right].$$
To prove (\ref{eq:perkfrj}), let us start by noting that if we set $z_e = \frac{x^{(t)}_e}{(1+\alpha) p}$ for any realized edge in $S$, we get $\sum_{v\in C} x^{(t)}_v = (1+\alpha)\sum_{v\in C}  \E[z_v]$ and proves (\ref{eq:perkfrj}). However, in our algorithm we decrease $z_e$ to zero for the edges around any bad vertex and get $z_v=0$ for any such vertex. Thus, using Chebyshev's inequality, we will show that for any $v$, probability of being bad is small and as a result reducing the fractional value of the edges around these vertices does not affect the expected size of our matching significantly.
  Since for any $v\in S$, we have  $\E[z_v] = x^{(t)}_v/ (1+\alpha)$, we get $$\Pr[z_v > x^{(t)}_v ] =  \Pr[z_v> \E[z_v](1+\alpha)] \leq \Pr[z_v - \E[z_v]> \alpha].$$
 Since any edge gets value of $x^{(t)}_e/(1+\alpha)p$ with probability  $p$ (if realized) and zero otherwise, for any edge $e \ni v$, we have:
  $$\Var[z_e] = \E[z_e^2] - \E[z_e]^2  =  p \left(x^{(t)}_e/(1+\alpha)p\right)^2 - \left(x^{(t)}_e/(1+\alpha)\right)^2 \leq \frac{\left(x^{(t)}_e\right)^2 (1-p)}{p} \leq \frac{t^2}{p}.$$
  Since edges are realized independently, this gives us $$\Var[z_v] = \sum_{e\ni v} \Var[z_v] \leq \frac{1}{t} \frac{t^2}{p} = \frac{t}{p}.$$
   Using Chebyshev's inequality, we have $$\Pr[z_v - \E[z_v]> \alpha] \leq \frac{\var[z_v]}{\alpha^2} \leq \frac{t}{p\alpha^2}.$$
  Now let us investigate the expected size of the fractional matching after decreasing $z_e$ to zero for any edge $e \ni v$ adjacent to a bad vertex $v$.
We have 
\begin{align}\label{eq:jahaheq}\E\left[\sum_{v\in V} z_v \right] \geq \frac{1}{1+\alpha}\sum_{v\in V} x^{(t)}_v(1-\frac{t}{p\alpha^2}).\end{align}
By setting $t= 0.25\alpha^2\epsilon p = \Theta(\tthreshold)$, we get 
 $$\E\left[\sum_{v\in V} z_v \right] \geq \frac{1-0.25\epsilon}{1+0.25\epsilon}\sum_{v\in V} x^{(t)}_v\geq \frac{1}{1+0.5\epsilon}\sum_{v\in V} x^{(t)}_v.$$
 This gives us (\ref{eq:perkfrj}) and completes the proof.
\end{proof}

\subsection{Bipartite Graphs: $(1+\epsilon)$-Approximation with $O_p(1)$ Per-Vertex Queries}

In this section, we will prove the following result:

\begin{theorem}\label{thm:bipartite1-eps}

For any constant $\epsilon > 0$ and constant $p \in (0, 1]$, any bipartite graph $G=(V, E)$ has a constant degree subgraph $Q$ where querying only the edges in $Q$ suffices to find $C \subseteq V$ such that
\begin{enumerate}[itemsep=0pt, topsep=5pt]
		\item $C$ is a vertex cover of $G_p$ with probability 1.
		\item The size of $C$ is in expectation at most $(1+\epsilon)$  times the minimum vertex cover of $G_p$.
		\item It is possible to find $Q$ and $C$ in polynomial time.
	\end{enumerate}
\end{theorem}

We start by the following lemma. While it is a folklore result, to be self-contained we also provide a proof for it.

\begin{claim}\label{cl:tobewritenjef}
Given $M_1$ and $M_2$, two  random matchings of graph $G$ with $\E[|M_2|] \geq \E[|M_1|]$, let $D$ denote their symmetric difference. For any $\epsilon \in (0, 1)$ graph $D$ contains at least $\E[|M_2|]- \E[|M_1|] - \frac{\epsilon\E[|M_2|]}{2}$ maximal paths of length at most $4/\epsilon+1$ who start and end with edges in $M_2$.
\end{claim}

\begin{proof}
Any maximal path of graph $D$ that starts and ends with edges in $M_2$ is also called an augmenting paths.
Since $D$ is the symmetric difference of $\E[|M_1|]$ and $\E[|M_2|]$ it has $\E[|M_2|] -\E[|M_1|]$ more edges from $M_2$ compared to $M_1$. As a result, it contains at least $\E[|M_2|] -\E[|M_1|]$ augmenting paths. Also, note that any augmenting paths longer than $4/\epsilon+1$, contains at least $ 2/\epsilon$ edges from $M_2$. Hence, there are most $\E[|M_2|]/(2/\epsilon)$ augmenting paths of length greater than $4/\epsilon+1$. As such, $D$ contains at least $$\E[|M_2|] -\E[|M_1|] - \frac{\epsilon\E[|M_2|]}{2}$$ 
augmenting paths of length at most  $4/\epsilon+1$.
\end{proof}

\begin{proof}[Proof of Theorem~\ref{thm:bipartite1-eps}.] 
	We provide a reduction to approximate stochastic matchings. Suppose that we have a stochastic matching algorithm $\mu_\delta$ that provides a $(1-\delta)$-approximation via $f(\delta, p)$ per-vertex queries -- such algorithm exists as proved in \cite{stoc20} for
	$$
		f(\delta, p) = \exp \Big(\exp\Big(\exp\Big(O(\delta^{-1}) \Big) \times \log\log p^{-1}\Big)\Big),
	$$ 
	and takes polynomial time to run for constant $\delta$ and $p$. We give a $(1+\epsilon)$-approximate stochastic minimum vertex cover algorithm that queries $f(\frac{\epsilon p^{2/\epsilon+2}}{4}, p) = O_{\epsilon, p}(1)$ edges per vertex. 
	
\smparagraph{The Reduction:} For $\delta = \frac{\epsilon p^{2/\epsilon+2}}{4}$, let $Q$ denote the set of edges queried by algorithm $\mu_\delta$ and let $S=E \setminus Q$. We claim that
querying set $Q$ and picking a minimum vertex cover of $(Q_P\cup S)$ is a ($1+\epsilon)$-approximate vertex cover algorithm. Let $\nu_{\epsilon}$ be the described algorithm. First, this clearly gives us a valid vertex cover of $G_p$ as it covers all the realized edges of $Q$ and all the edges that are not queried (i.e., edges in $S$). Thus, to complete the proof we need to show  
\begin{align}\label{eq:liwejw} \E[\nu(Q_P\cup S)] \leq (1+\epsilon) \cdot \E[\nu(G_p)]. \end{align} 
For the sake of contradiction, we assume that this inequality does not hold and show that it implies $\E[|\mu(Q_p)|] < (1-\delta)\E[|\mu(G_p)|]$, contradicting that $\mu_{\delta}$ is a $(1-\delta)$-approximate stochastic matching algorithm. Note that since $G$ is bipartite, the size of its maximum matching and minimum vertex cover are equal by K\"onig's theorem. This implies that if (\ref{eq:liwejw}) does not hold, then we have  
\begin{align}\label{eq:oiwebfjhew}
\E[\mu(Q_P\cup S)] > (1+\epsilon) \E[\mu(G_p)].
\end{align}
 We will show that in this case, matching $\mu(Q_p)$ can be augmented using edges in $S_p$ to a matching whose size is larger than that of $\mu(G_p)$ in expectation, which is a contradiction. Let $D$ denote the symmetric difference of $M_1:=\MM{Q_p}$ and $M_2:=\MM{Q_P\cup S}$ where $\MM{\cdot}$ returns an arbitrary maximum matching. Namely, $D$ contains an edge $e$ if it is in exactly one of these matchings. By Claim~\ref{cl:tobewritenjef},  graph $D$ in expectation contains at least 
 $\E[|M_2|] -\E[|M_1|] - \frac{\epsilon\E[|M_1|]}{2}$
 maximal paths of length at most $ \frac{4}{\epsilon}+1$ which start and end with edges in $M_2$ (i.e., augmenting paths), where 
  \begin{flalign*}\E[|M_2|] -\E[|M_1|] - \frac{\epsilon\E[|M_1|]}{2} &\geq \E[|\mu(Q_P\cup S)|]-\E[|\mu(Q_p)|] - \frac{\epsilon\E[|\mu(Q_p)|]}{2}\\&\geq \E[|\mu(Q_P\cup S)|]-\E[|\mu(G_p)|] - \frac{\epsilon\E[|\mu(G_p)|]}{2} \\ &  \hspace{-1.2mm} \stackrel{(\ref{eq:oiwebfjhew})}{>}  (1+\epsilon)\E[|\mu(G_p)|] -\E[|\mu(G_p)|] - \frac{\epsilon\E[|\mu(G_p)|]}{2} \\ &>  \frac{\epsilon \E[|\mu(G_p)|]}{2}.  \end{flalign*}
 Consider one of these augmenting paths $P$. If all the edges of $P$ are in $G_p$, flipping the membership of its edges in  $M_1$ increases the size of this matching  by one. Note that any of these augmenting paths has at most $ \frac{2}{\epsilon}+2$ edges from $S$ and these edges are realized (are in $G_p$) independently  with probability $p$. As a result, each one of these paths is in $G_p$ with probability at least $p^{ \frac{2}{\epsilon}+2}.$ Since we have at least $\frac{\epsilon \E[|\mu(G_p)|]}{2}$ of these paths, applying all of them on $M_1$ results in increasing its expected size by at least $$p^{ \frac{2}{\epsilon}+2} \cdot \frac{\epsilon \E[\mu(G_p)]}{2}= 2\delta\E[\mu(G_p)].$$
This means that the resulting matching has size at least $$\E|M_1| + 2\delta\E[\mu(G_p)] \geq (1-\delta) \E[\mu(G_p)] + 2\delta\E[\mu(G_p)] > \E[\mu(G_p)],$$ which is a contradiction by  the fact that $\E[\mu(G_p)]$ is an upper bound  for $\E|M_1|$.
\end{proof}

\section{Lower Bounds}\label{sec:lowerbounds}

In this section, we prove several lower bounds for both the stochastic vertex cover problem and also the stochastic matching problem. Below we state these results as a series of theorems, and give their proofs later in the section.

\begin{theorem}\label{theorem:lowekmwj}
For any constant $p<1$, there exist an $n$-vertex bipartite stochastic graph $G$ with realization probability $p$ such that finding an exact minimum vertex cover or an exact maximum matching of $G$ with any constant probability requires querying $\Omega(\frac{n^2}{\log^2 n})$ edges of this graph.
\end{theorem}

\begin{theorem}\label{theorem:ourfiuer}
	Finding a maximal matching of $n$-vertex stochastic graphs with a constant realization probability $p\in (0, 1)$ requires $\Omega(n \log_b n)$ total queries for $b = \frac{1}{1-p}$.
\end{theorem}

\begin{theorem}\label{theorem:oieoi2owa}
There are absolute constants $p_0, \epsilon_0 \in (0, 1)$ such that for any $p \leq p_0$ and $\epsilon \leq \epsilon_0$, finding a $(1-\epsilon)$-approximate maximum matching for a bipartite stochastic graph $G_p$, requires querying a subgraph of maximum degree $\Omega(\frac{1}{\epsilon p})$.
\end{theorem}

\begin{theorem}\label{theorem:elijfoew}
	There are absolute constants $p_0, \epsilon_0 \in (0, 1)$ such that for any $p \leq p_0$ and $\epsilon \leq \epsilon_0$, finding a $(1+\epsilon)$-approximate minimum vertex cover for a bipartite stochastic graph $G_p$, requires querying a subgraph of maximum degree $\Omega(\frac{1}{\epsilon p})$.
\end{theorem}

\begin{theorem}\label{theorem:oirwfno3iw}
	Finding a constant approximation of minimum vertex cover of $n$-vertex stochastic bipartite graphs with realization probability $p$ requires $\Omega(\frac{n}{p})$ total queries.
\end{theorem}

A graph that is particularly useful for our lower bounds is illustrated in Figure~\ref{fig:LB} and  defined formally in Definition~\ref{def:kjierrfjk}.

\begin{figure}[hbt]
  \centering
  \includegraphics{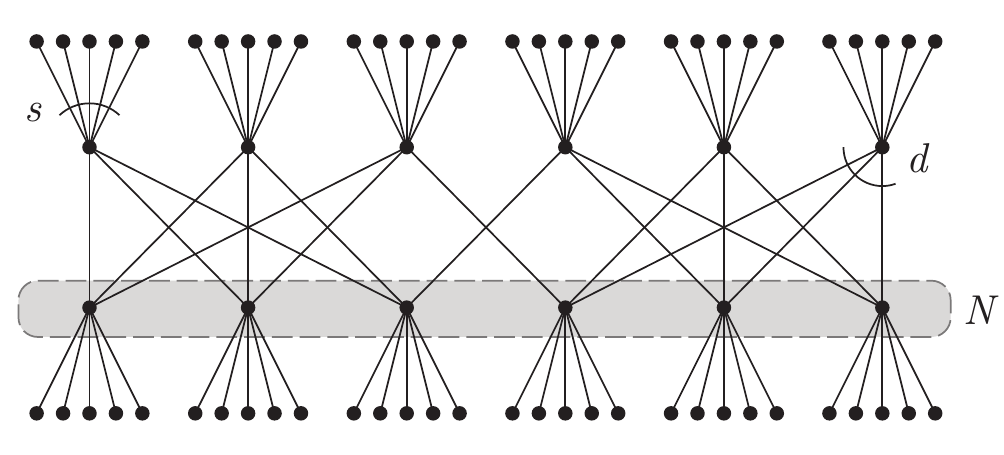}
  \caption{An example of graph $S(d, s, N)$ formalized in Definition~\ref{def:kjierrfjk} for $d = 3, s = 5, $ and $N = 6$.}
  \label{fig:LB}
\end{figure}
\begin{definition}[$S(d, s, N)$-graphs --  Figure~\ref{fig:LB}] \label{def:kjierrfjk} For positive integers $d, s, $ and $N$, an $S(d, s, N)$-graph is defined on $2N(s+1)$ vertices. Of these nodes, $2N$ form an induced $d$-regular bipartite graph $B$, with $N$ nodes in each partition. In addition, each vertex of $B$ is also connected to $s$ vertices outside $B$, each with degree exactly 1. We use $S$ to denote the set of vertices outside $B$.
\end{definition}
In proving our theorems we use,  two essential properties of $S(d, s, N)$-graphs which we state below as Lemma~\ref{lemma:matchhoiefhk} and Lemma~\ref{lemma:vertexhoiefhk}.

\begin{lemma}\label{lemma:matchhoiefhk}

Given a stochastic $S(d, s, N)$-graph $G$ with realization probability $p$, consider its subgraphs  $B$ and $S$ from Definition~\ref{def:kjierrfjk} and let $H$ be the induced subgraph of vertices in $B$ that do not have any realized edges to vertices in $S$. If $M$ is a $(1-\epsilon)$-approximate maximum matching of $G_p$, it contains a matching of expected size at least $|\mu(H_p)| - 2N\epsilon$ from $H$.
	
\end{lemma}

\begin{proof}
Let $M_{apx}$ be an arbitrary $(1-\epsilon)$-approximate maximum matching of $G_p$. To prove this lemma, we will show $$|M_{apx}\cap H| \geq |\mu(H_p)| - 2N\epsilon.$$
We start by giving a lower bound for $|M_{apx}|$. We do so by constructing a matching of graph $G_p$ which we denote by $M_1$. In $M_1$, we match any vertex $v\in B$ that has a realized edge to a vertex in $S$ (i.e., any vertex in $B$ that is not in $H$) using one of these edges. We also add a maximum matching of subgraph $H_p$ to $M_1$. Let  $V_H$ and $V_B$ respectively denote the vertex sets of graphs $H$ and $B$. Observe that we have $|M_1|  = |V_B-V_H| + |\mu(H_p)|$. Knowing that $M_{apx}$ is a $(1-\epsilon)$-approximate maximum matching, gives us:
$$|M_{apx}|\geq  (1-\epsilon)( |V_B| - |V_H| + |\mu(H_p)|) \geq |V_B| - |V_H| + |\mu(H_p)| - \epsilon2N .$$
We complete the proof by noting that the expected size of the maximum matching on $G_p\setminus H_p$ is upper bounded by $|V_B| - |V_H|$ since any edge in $G_p\setminus H_p$ has an endpoint in $V_B\setminus V_H$ and $|V_B\setminus V_H| = |V_B| - |V_H|.$ This completes the proof since it indicates that the remaining $|\mu(H_p)| - \epsilon2N$ edges of matching $M_{apx}$ come from $H_p$.
\end{proof}

\begin{lemma}\label{lemma:vertexhoiefhk}

Given a stochastic $S(d, s, N)$-graph $G$ with realization probability $p$, consider its subgraphs  $B$ and $S$ from Definition~\ref{def:kjierrfjk} and let $H$ be the induced subgraph of vertices in $B$ that do not have any realized edges to vertices in $S$. Any ($1+\epsilon$)-approximate minimum vertex cover of $G_p$, in expectation includes at most $|\nu(H_p)| + 2N\epsilon$ vertices  from $H$.
	
\end{lemma}
\begin{proof}
Let $\nu_{apx}$ be an arbitrary $(1+\epsilon)$-approximate minimum vertex cover of  $G_p$. To prove this lemma, we will show $$|\nu_{apx}\cap H_p| \leq \nu(H_p) +2N\epsilon.$$ 
We start by giving an upper bound for $|\nu_{apx}|$. We do so by constructing a vertex cover of graph $G_p$ which we denote by $\nu_1$. This vertex cover includes any vertex $v\in B$ that has a realized edge to a vertex in $S$ (i.e., any vertex in $B$ that is not in $H$). This covers all the edges that are not in $H$. Thus, we complete $\nu_1$ by adding a minimum vertex cover of $H_p$ to it.
 Let  $V_H$ and $V_B$ respectively denote the vertex sets of graphs $H$ and $B$. Observe that we have $|\nu_1|  = |V_B-V_H| + |\nu(H_p)|$. As a result of $\nu_{apx}$ being a $(1+\epsilon)$-approximate minimum vertex cover we get:
$$|\nu_{apx}| \leq   (1+\epsilon)( |V_B| - |V_H| + |\nu(H_p)|) \leq |V_B| - |V_H| + |\nu(H_p)| + \epsilon2N .$$
We complete the proof by noting that $|V_B| - |V_H|$ is a lower bound for the number of vertices in $V_B\setminus V_H$ that are in vertex cover since any vertex in $V_B\setminus V_H$  is connected to at least a degree-one vertex.
This concludes the proof as it indicates that at most $|\mu(H_p)| + \epsilon2N$ vertices of $\nu_{apx}$ can come from $H_p$.
\end{proof}

\subsection{Proof of Theorem~\ref{theorem:lowekmwj}}
We start the proof by the following Lemma.
\begin{lemma} \label{lemma:kjnrjf}
For any constant $p<1$, there exist an $n$-vertex bipartite stochastic graph $G$ with realization probability $p$ and a constant $c\geq 0$ such that finding an exact maximum matching or an exact minimum vertex cover of $G$ with probability at least $c$ requires querying $\Omega(\frac{n^2}{\log^2 n})$ edges of this graph.
\end{lemma}
\begin{proof}
To prove this lemma, we start by an $n$-vertex graph $G$, a $S(N, s, N)$-graph for $s=\log_{1-p} 1/N$. We then prove the existence of a constant $c$ for which the statement of the lemma holds. Particularly, we will show that it is not possible to find an exact maximum matching/minimum vertex cover of $G_p$ with probability at least $c$ using only $o(n^2/\log^2 n)$ queries. 

Consider subgraphs  $B$ and $S$  of $G$ from Definition~\ref{def:kjierrfjk} and let $H$ be the induced subgraph of vertices in $B$ that do not have any realized edges to vertices in $S$. Based on Lemma~\ref{lemma:matchhoiefhk}, any maximum matching of $G_p$ should include a maximum matching of $H_p$.  Also, based on Lemma~\ref{lemma:vertexhoiefhk}, any minimum vertex cover of $G_p$, includes exactly $|\nu(H_p)|$ vertices from $H_p$. Thus, to be able to find an exact minimum vertex cover/maximum matching of $G_p$ we need to find a minimum vertex cover/maximum matching of $H_p$. Thus, we will focus on graph $H_p$.

Since $s = \Theta (\log N)$, we have $n = \Theta(N\log N)$ and $N = \Omega(n/\log n)$ (based on Definition~\ref{def:kjierrfjk}). This implies that if at most $o(n^2/\log^2 n)$ edges are queried from $G$, then any edge $e$ chosen uniformly at random from $B$ (defined in Definition~\ref{def:kjierrfjk}) is queried with probability $o(1)$ (since $B$ has $N^2 = \Omega(n^2/\log^2 n)$ edges). Therefore, it suffices to show that if any random edge from $B$ is queried with probability $o(1)$, then it is not possible to find an exact maximum matching/minimum vertex cover of $H_p$ with probability at least $c$ for a constant $c$.
 
 We claim that with a constant probability subgraph $H_p$ has only one vertex in each part. Let us denote this event by $E$ and compute its probability. Since each vertex $v \in B$ is in $H$ independently with probability $(1-p)^s = 1/N$, we have
   $$ \Pr[E] = (N (1/N) (1-1/N)^{N-1} )^2 \geq 1/e^2.$$
 When event $E$ happens; i.e., $H$  has only one vertex in each part, to be able to find its maximum matching, if edge $(u,v)$ is realized it also should be queried. This edge is realized with probability $p$, however, as discussed above, probability of this edge being queried is $o(1)$.  This implies that with probability at least $1/e^2 (p- o(1))$, which is a constant, the queried edges do not contain a maximum matching of $G_p$. Similarly, if edge $(u,v)$ is not queried, to have a valid vertex cover, we need to put one of its endpoints in the vertex cover, however, this vertex cover is not minimum if edge $(u,v)$ is not realized. This event happens with probability $1/e^2(1-p-o(1))$ which is again a constant. As a result, setting $c= min(1/e^2(1-p-o(1)), 1/e^2(p- o(1)))$ completes the proof.
\end{proof}

Given Lemma~\ref{lemma:kjnrjf}, we are now ready to prove Theorem~\ref{theorem:lowekmwj}.

\begin{proof}[Proof of Theorem~\ref{theorem:lowekmwj}]
We use proof by contradiction. We first assume the existence of a constant $\alpha$ such that given any $n$-vertex stochastic graph $G$, it is possible to find an exact minimum vertex cover/maximum matching of this graph with probability at least $\alpha$ using only $o(n^2/\log^2 n)$ queries. We then show that this results in a contradiction. Based on Lemma~\ref{lemma:kjnrjf}, for any $N$ there exists an $N$-vertex bipartite stochastic graph $G'$ with realization probability $p$ and a constant $c\geq 0$ such that finding an exact minimum vertex cover or an exact maximum matching of $G'$ with probability at least $c$ requires $\Omega(N^2/\log^2 N)$ queries. Define graph $G$ to include $s=\lceil\log_{1-c} (1-\alpha) \rceil$ copies of $G'$ with $N=n/s$. Since $n=\Theta(N)$, querying $o(\frac{n^2}{\log^2 n})$ edges of $G$, means queries  $o(\frac{N^2}{\log^2 N})$ edges of each subgraph. Observe that finding an exact maximum matching/minimum vertex cover of $G$ means finding  an exact maximum matching/minimum vertex cover of all these $s$ subgraphs.  This implies that the queried edges of $G$ contain an exact maximum matching/minimum vertex cover of $G_p$ with probability at most $1- (1-c)^s \leq 1-(1-\alpha) \leq \alpha$. We conclude the proof of this theorem by noting that this is a contradiction with the initial assumption.
\end{proof}

\subsection{Proof of Theorem~\ref{theorem:ourfiuer}}
\begin{proof}[Proof of Theorem~\ref{theorem:ourfiuer}]
Consider an $n$-vertex clique $G$ with realization probability $p$. It is known that for any constant $p\in (0,1)$, the expected size of the maximum independent set in $G_p$  is $\Theta(\log_b n)$ for $b=1/(1-p)$. (See Theorem 7.3 in Book~\cite{frieze2016introduction}.) This is useful for giving a lower bound for the size of any maximal matching $M$ of $G_p$ since vertices that are not  in $M$ form an independent set. This implies $2|M|\geq n - |\text{MIS}(G_p)|$ where $\text{MIS}(G_p)$ is the maximum independent set of $G_p$. Let $Q$ be the subgraph we choose to query. If $Q_p$ contains a maximal matching then its maximum independent set  should not be larger than that of $G_p$. Therefore, to complete the proof, it suffices to show that if \begin{equation} \label{eq:lijfhhra}|Q_p| = o(n \log_b n) < \frac{n \log_b n}{4},\end{equation}
 then $\E[|\text{MIS}(Q_p)|] = \omega(\log_b n)$. To prove this, we will show that the expected number of singleton vertices in $Q_p$ is $\omega(\log_b n)$.
Note that based on (\ref{eq:lijfhhra}), at least half of vertices have degree at most $\frac{\log_b n}{2}$ in $Q$. For any such vertex $v$ we have:
$$\Pr[v \text{ is singleton in } Q_p] \geq (1-p)^{\frac{\log_b n}{2}} = (1/b)^{\frac{\log_b n}{2}} = (\frac{1}{n})^\frac{1}{2}.$$ For any constant $p\in (0,1)$, this gives us:
$$\E[|\text{MIS}(Q_p)|] \geq \E[\text{number of singleton vertices in } Q_p] \geq \frac{n}{2}\cdot (\frac{1}{n})^\frac{1}{2} = \frac{\sqrt{n}}{2} = \omega (\log_b n),$$
and completes the proof.\end{proof}

\subsection{Proof of Theorem~\ref{theorem:oieoi2owa}}
\begin{proof}[Proof of Theorem~\ref{theorem:oieoi2owa}]
Consider the bipartite graph $G = S(N, s, N)$ for $s=\lfloor \log_{1-p} 10\epsilon  \rfloor$ (given that $10\epsilon \leq 1-p$ holds for small enough values of $\epsilon$ and $p$), and let $n$ be the total number of vertices in $G$.  To prove this theorem, we show that finding a $(1-\epsilon)$-approximate maximum matching of $G_p$ requires $\Omega(\frac{1}{\epsilon p})$ queries for a vertex. Consider subgraphs  $B$ and $S$  of $G$ from Definition~\ref{def:kjierrfjk} and let $H$ be the induced subgraph of vertices in $B$ that do not have any realized edges to vertices in $S$. Based on Lemma~\ref{lemma:matchhoiefhk}, any  ($1-\epsilon$)-approximate maximum matching of $G_p$, includes a matching of expected size at least $|\mu(H_p)| - 2N\epsilon$ from $H$. Let $X_1$ and $X_2$ with be the set of vertices in two parts of the bipartite graph $H$ with $|X_1|\leq |X_2|$.
Note that the expected number of vertices in each part is 
$$(1-p)^{\lfloor \log_{1-p} 10\epsilon  \rfloor} N \geq 10\epsilon N;$$
thus, $\E[|X_1|] \geq 5\epsilon N$.  
It is well-known that such random graph $H_p$ that has an edge between any two vertices in $X_1$ and $X_2$ with a constant probability $p$ has a matching of size $|X_1|$ with high probability. However, for the sake of this proof we only use the fact that a matching of size $|X_1|$ exists with probability at least $\frac{4}{5}$. This gives us $  \E[|\mu{(H_p)}|] \geq 4\epsilon N$.  As a result, to find a matching of size at least $|\mu{(H_p)}| - 2\epsilon N$ in $H_p$, any randomly chosen vertex $v\in X_1$ should have at least one queried edge in $H_p$ with probability at least $\frac{1}{2}$. Let $g_v$ be the expected number of edges queried for vertex $v$. Since the other endpoint of any of these queried edges is in $H_p$ with probability  $$(1-p)^{\lfloor \log_{1-p} 10\epsilon  \rfloor} \leq 10\epsilon/(1-p),$$ and $e$ is realized with probability $p$, we should have $\frac{10\epsilon pg_v}{1-p} \geq 1/2$ which for any $p\leq 0.99$ gives us $g_v \geq \frac{1-p}{20\epsilon p} = \Omega(\frac{1}{\epsilon p})$ and concludes the proof.
\end{proof}

\subsection{Proof of Theorem~\ref{theorem:elijfoew}}
\begin{proof}[Proof of Theorem~\ref{theorem:elijfoew}]
\newcommand{\avar}[0]{96}
\newcommand{\bvar}[0]{2^{15}}
\newcommand{\cvar}[0]{2^{16}}
Consider the bipartite graph $G = S(d, s, N)$ for $$s=\lfloor \log_{1-p} \avar \epsilon \rfloor  \quad \quad \text{ and } \quad \quad d = 1/(\bvar \epsilon p),$$ and let $n$ be the total number of vertices in $G$. 
To prove this theorem, we show that it is not possible to find a $(1+\epsilon)$-approximate minimum vertex cover of $G_p$ using only $o(1/\epsilon p)$ queries per vertex. 
 We use proof by contradiction. We assume the existence of a $(1+\epsilon)$-approximate minimum vertex $\nu_{apx}$, achieved by querying only $o(1/\epsilon p)$ edges per vertex and show that it results in a contradiction.  Consider subgraphs  $B$ and $S$  of $G$ from Definition~\ref{def:kjierrfjk}, and let $H$ be the induced subgraph of vertices in $B$ that do not have any realized edges to vertices in $S$. Based on Lemma~\ref{lemma:vertexhoiefhk}, any  ($1+\epsilon$)-approximate minimum vertex cover of $G_p$, includes at most $|\nu(H_p)| + 2N\epsilon$ vertices from $H$, which means $$|\nu_{apx} \cap H| \leq|\nu(H_p)| + 2N\epsilon.$$ Let us define a subgraph $H'$ of $H$ to contain all the edges of $H$ that are not queried. To be a valid vertex cover, $\nu_{apx}$  should include a minimum vertex cover of $H'$; i.e., $$|\nu(H')| \leq |\nu_{apx} \cap H|.$$  Thus, to achieve a contradiction, it suffices to prove that the following equation does not hold if any vertex has $d'=1/(\cvar \epsilon p) +1$  edges in $B$ that are not queried: \begin{equation}|\nu(H')| \leq |\nu(H_p)| + 2N\epsilon. \label{eq:kirfbjher}\end{equation}
Let us assume w.o.l.g., that exactly $d'$ edges are not queried for any vertex in $B$ as it only decreases the size of the minimum vertex cover. We start by giving a lower bound for $|\nu(H')|$. We do so by constructing a fractional matching $M$  on $H'$. Since by weak duality, the minimum vertex cover of a graph is at least as large as any fractional matching of the graph, this will also be a lower bound for $|\nu(H')|$.  Define $x:=(d'-1)(1-p)^s+1$.
 For any edge $e$, let $m_e$ be the fractional value we assign to edge $e$ in matching $M$. We set $m_e := 1/x$  for any edge $e\in H'_p$ iff both its endpoints have degree at most $x$ in $H'_p$ and zero otherwise. $M$ is obviously a valid fractional matching since the sum of values assigned to the edges around any vertex is at most one. 
 Also, for any edge $e = (u, v)\in B$, we have \begin{equation}\label{eq:rfkjkr} \E[m_e] =  \Pr[e\in H']\cdot\Pr[d_{v, H'}\leq x \;|\; e\in H']\cdot\Pr[d_{u, H'}\leq x \;|\; e\in H'],\end{equation} where $d_{v, H'}$ is the degree of vertex $v$ in graph $ H'$.
 To compute $\Pr[e\in H']$, note that each endpoint of this edge is in $H$ with probability $(1-p)^s$. Combining this with $s=\lfloor \log_{1-p} \avar \epsilon \rfloor$ gives us 
 \begin{equation}\label{eq:ofnrfg}(\avar\epsilon)^2 \leq \Pr[e\in H'].\end{equation}
 Consider any vertex $v\in H'$ and any of its edges $e\in H'$, and let $d_{v, e, H'}$ denote the degree of vertex $v$ in  $H'\setminus \{e\}$.
 Since $v$ has $d'-1$ edges in $B \setminus \{e\}$, we have 
  $\E[ d_{v, e, H'}]= (d'-1)(1-p)^s = x-1$. Moreover, we claim that as a result of  $d_{v, e, H'}$ being sum of independent Bernoulli random variables, we have $\Pr[ d_{v, e, H'} > x-1] < 0.75$. This can be achieved using a simple application of Chernoff bound, and implies
 $$\Pr[ d_{v, H'} \leq x \;|\; e\in H'] \geq \frac{1}{4}.$$ Combining this with (\ref{eq:ofnrfg}) and (\ref{eq:rfkjkr}) gives us:
$$|M| \geq Nd' \frac{1}{x}(\avar \epsilon)^2 \frac{1}{16} \geq \frac{N\avar^2\epsilon^2}{\cvar \epsilon p x 16}.$$
Observe that for a small enough $p$, we have $$x= (d'-1)(1-p)^s+1 \leq  \frac{\avar}{\cvar p (1-p)} +1 \leq  \frac{2\cdot \avar}{\cvar p (1-p)}.$$ This implies:
\begin{equation}|M| \geq \frac{N\avar^2\epsilon^2}{16 \cdot \cvar \epsilon p (1-p)}\frac{\cvar p(1-p)}{2\cdot \avar} \geq \frac{N\avar\epsilon}{32}.\end{equation}
After finding a lower bound for $|M|$ which is also a lower bound for $|\nu(H')|$ we need to find an upper bound for $|\nu(H_p)|$. Recall that  $H$ is the induced subgraph of vertices in $B$ that do not have any realized edges to vertices in $S$. To give an upper bound for $|\nu(H_p)|$, we start by computing the expected number of edges in this graph which we denote by $|H_p|$. Each edge $e\in B$ is in $H_p$ with probability $p(1-p)^2s \leq p \avar^2 \epsilon^2/(1-p)^2$. Since $B$ has $N/(\bvar \epsilon p)$ edges, for a large enough $(1-p)$ we get $$|H_p| \leq \frac{N\avar^2\epsilon}{\bvar(1-p)^2}\leq \frac{2N\avar^2\epsilon}{\bvar}.$$ This is also an upper bound for $|\nu(H_p)|$ since minimum vertex cover is not larger than the total number of edges.
Based on (\ref{eq:kirfbjher}), to complete the proof we need to show that 
$$\frac{N\avar\epsilon}{32} \geq \frac{2N\avar^2\epsilon}{\bvar} + 2N\epsilon.$$
If we simplify both sides of the equation we get 
$$\frac{N\avar\epsilon}{32} = 3N\epsilon \geq 2.6 N\epsilon \geq \frac{2N\avar^2\epsilon}{\bvar} \geq + 2N\epsilon,$$
 which concludes the proof.
\end{proof}

\subsection{Proof of Theorem~\ref{theorem:oirwfno3iw}}
\begin{proof}[Proof of Theorem~\ref{theorem:oirwfno3iw}]
Suppose we want to get a $c$-approximate vertex cover for a $c>1$. For any $p<\frac{1}{3c}$, we construct a graph $G$ for which finding such an approximation needs at least $\Omega(n/p)$ queries. Graph $G$ is an arbitrary $d$-regular bipartite graph with $d =  \frac{1}{3cp}$.  (Here we assume for simplicity that $d$ is an integer.) In $G_p$, any vertex $v$ is a singleton with probability at least $1-dp$ which means $dnp/2$ is an upper bound for the expected size of the vertex cover. 

On the other hand, if we do not query an edge, we have to put one of its endpoints in the vertex cover. We claim that to get a $c$-approximate at least half of the edges ($dn/4$) should be queried. Since any vertex can cover at most $d$ edges, not querying  $dn/4$ edges results in a vertex cover of size at least $n/4$. Thus, if at most $dn/4$ edges are queried, the approximation factor of the algorithm would be at least $\frac{n/4}{dnp/2} = \frac{1}{2dp} = \frac{3c}{2} \geq c$. To conclude, to be able to get a $c$-approximate at least $\frac{n}{6cp} = \Omega(n/cp)$ (half the number of edges) queries are needed.
\end{proof}

\subsection{Why Random Queries Do Not Work}\label{sec:randomqueries}
One simple approach that, in the first sight, might seem appealing  for finding  approximate vertex covers of a stochastic graph $G$ is simply picking a random subset of edges of any vertex to form a subgraph $Q$, query these edges and find the vertex cover based on that. In this section, we will show that this approach, in fact, does not have a good performance for some instances of the problem. Namely, we show that there exists a graph $G$ such that querying a subgraph $Q$ of $G$ obtained via sampling $s= o(n)$ edges per any vertex does not give us better than a $(\frac{1}{p} - o(1))$-approximate vertex cover of $G_p$. Note that to find a valid vertex cover, any edge in $G$ that is not queried should be covered; i.e., edges in $S := G \setminus Q$. Therefore, the final vertex cover found via this sampling technique has size at least $|\nu(S \cup Q_p)|$ in expectation.  

We construct an $n$-vertex bipartite graph $G$ as follows. For a number $N$ with $N=o(n)$ and $N=\omega(s)$, graph $G$ contains a complete bipartite graph $H$ of $N$ vertices and a matching $M$ of size $(n-N)/2$. Moreover, for any $e=(u,v)\in M$,  vertex $u$ has edges to all the vertices in one part of $H$ and $v$ has edges to all the vertices in the other part of $H$. Note that such a number $N$ exists since $s = o(n)$. Let us first give an upper bound for $|\nu (G_p)|$ by constructing a vertex cover $C$ of this graph. $C$ contains all the vertices of graph $H$ and an endpoint of any edge in $M$ that is realized (an endpoint of any edge in $M \cup G_p$). This set clearly covers all the edges of $G_p$. Since any edge in $M$ is realized with probability $p$, we have 
\begin{align} \label{eq:wfgpwdpqgf} \E[|\nu(G_p)|] \leq \E[|C|]= N + \frac{p(n-N)}{2} \leq pn/2  + N/2 = pn/2 +o(n).\end{align}
Having established this upper bound, the next step is to give a lower bound for $|\nu(S \cup Q_p)|$. Consider an edge  $e = (u, v)\in M$, and let us compute $\Pr[e\in Q]$. Both end pints of this edge have degree $N/2 +1$ in graph $G$. Since each vertex randomly chooses $s$ edges, for the probability of $e$ being sampled we have: 
$$\Pr[e\in Q] \leq \frac{2 s}{N/2 +1} = o(1).$$
As a result the expected number of edges sampled from matching $M$ is at most $(1-o(1))|M|.$ Since any edge that is not queried (sampled) should be covered, any valid vertex cover of $S\cup Q_p$ should contain at least one end point of edges in $M$ that are not sampled. Hence, we have  $$\E[|\nu(S \cup Q_p)|] \geq (1-o(1))|M|= (1-o(1))(n-N)/2  = n/2 - o(n).$$ Combining this with (\ref{eq:wfgpwdpqgf}) gives us:
$$\frac{\E[|\nu(S \cup Q_p)|]}{\E[|\nu(G_p)|]} \geq\frac{n/2-o(n)}{pn/2 + o(n)} = \frac{1}{p} - o(1). $$
This proves our claim as it means that by randomly sampling $o(n)$ edges for any vertex in graph $G$, one cannot construct a valid vertex cover of $G_p$ with an approximation ratio smaller than $(1/p- o(1))$.

\bibliographystyle{plain}
\bibliography{refs}

\begin{thebibliography}{10}

\bibitem{sosa19}
Sepehr Assadi and Aaron Bernstein.
\newblock {Towards a Unified Theory of Sparsification for Matching Problems}.
\newblock In {\em 2nd Symposium on Simplicity in Algorithms, SOSA@SODA 2019,
  January 8-9, 2019 - San Diego, CA, {USA}}, pages 11:1--11:20, 2019.

\bibitem{AKL16}
Sepehr Assadi, Sanjeev Khanna, and Yang Li.
\newblock {The Stochastic Matching Problem with (Very) Few Queries}.
\newblock In {\em Proceedings of the 2016 {ACM} Conference on Economics and
  Computation, {EC} '16, Maastricht, The Netherlands, July 24-28, 2016}, pages
  43--60, 2016.

\bibitem{AKL17}
Sepehr Assadi, Sanjeev Khanna, and Yang Li.
\newblock {The Stochastic Matching Problem: Beating Half with a Non-Adaptive
  Algorithm}.
\newblock In {\em Proceedings of the 2017 {ACM} Conference on Economics and
  Computation, {EC} '17, Cambridge, MA, USA, June 26-30, 2017}, pages 99--116,
  2017.

\bibitem{focs20}
Soheil Behnezhad and Mahsa Derakhshan.
\newblock Stochastic weight matching: $(1-\epsilon$)-approximation.
\newblock In {\em Foundations of Computer Science (FOCS 20), to appear}, 2020.

\bibitem{sagt19}
Soheil Behnezhad, Mahsa Derakhshan, Alireza Farhadi, MohammadTaghi Hajiaghayi,
  and Nima Reyhani.
\newblock {Stochastic Matching on Uniformly Sparse Graphs}.
\newblock In {\em Algorithmic Game Theory - 12th International Symposium,
  {SAGT} 2019, Athens, Greece, September 30 - October 3, 2019, Proceedings},
  pages 357--373, 2019.

\bibitem{stoc20}
Soheil Behnezhad, Mahsa Derakhshan, and MohammadTaghi Hajiaghayi.
\newblock {Stochastic Matching with Few Queries: {$(1-\epsilon)$}
  Approximation}.
\newblock In {\em Proceedings of the 52nd Annual {ACM} {SIGACT} Symposium on
  Theory of Computing, {STOC} 2020, to appear}, 2020.

\bibitem{soda19}
Soheil Behnezhad, Alireza Farhadi, MohammadTaghi Hajiaghayi, and Nima Reyhani.
\newblock {Stochastic Matching with Few Queries: New Algorithms and Tools}.
\newblock In {\em Proceedings of the Thirtieth Annual {ACM-SIAM} Symposium on
  Discrete Algorithms, {SODA} 2019, San Diego, California, USA, January 6-9,
  2019}, pages 2855--2874, 2019.

\bibitem{BR18}
Soheil Behnezhad and Nima Reyhani.
\newblock {Almost Optimal Stochastic Weighted Matching with Few Queries}.
\newblock In {\em Proceedings of the 2018 {ACM} Conference on Economics and
  Computation, Ithaca, NY, USA, June 18-22, 2018}, pages 235--249, 2018.

\bibitem{blumetal}
Avrim Blum, John~P. Dickerson, Nika Haghtalab, Ariel~D. Procaccia, Tuomas
  Sandholm, and Ankit Sharma.
\newblock {Ignorance is Almost Bliss: Near-Optimal Stochastic Matching With Few
  Queries}.
\newblock In {\em Proceedings of the Sixteenth {ACM} Conference on Economics
  and Computation, {EC} '15, Portland, OR, USA, June 15-19, 2015}, pages
  325--342, 2015.

\bibitem{blumetalOR}
Avrim Blum, John~P. Dickerson, Nika Haghtalab, Ariel~D. Procaccia, Tuomas
  Sandholm, and Ankit Sharma.
\newblock {Ignorance Is Almost Bliss: Near-Optimal Stochastic Matching with Few
  Queries}.
\newblock {\em Operations Research}, 68(1):16--34, 2020.

\bibitem{colbourn1987combinatorics}
Charles~J Colbourn.
\newblock {\em The combinatorics of network reliability}.
\newblock Oxford University Press, Inc., 1987.

\bibitem{dubhashi1996balls}
Devdatt~P Dubhashi and Desh Ranjan.
\newblock Balls and bins: A study in negative dependence.
\newblock {\em BRICS Report Series}, 3(25), 1996.

\bibitem{frieze2016introduction}
Alan Frieze and Micha{\l} Karo{\'n}ski.
\newblock {\em Introduction to random graphs}.
\newblock Cambridge University Press, 2016.

\bibitem{DBLP:conf/soda/GoemansV04}
Michel~X. Goemans and Jan Vondr{\'{a}}k.
\newblock Covering minimum spanning trees of random subgraphs.
\newblock In {\em Proceedings of the Fifteenth Annual {ACM-SIAM} Symposium on
  Discrete Algorithms, {SODA} 2004, New Orleans, Louisiana, USA, January 11-14,
  2004}, pages 934--941, 2004.

\bibitem{DBLP:journals/rsa/GoemansV06}
Michel~X. Goemans and Jan Vondr{\'{a}}k.
\newblock Covering minimum spanning trees of random subgraphs.
\newblock {\em Random Struct. Algorithms}, 29(3):257--276, 2006.

\bibitem{DBLP:journals/siamcomp/GuoJ19}
Heng Guo and Mark Jerrum.
\newblock A polynomial-time approximation algorithm for all-terminal network
  reliability.
\newblock {\em {SIAM} J. Comput.}, 48(3):964--978, 2019.

\bibitem{joag1983negative}
Kumar Joag-Dev and Frank Proschan.
\newblock Negative association of random variables with applications.
\newblock {\em The Annals of Statistics}, pages 286--295, 1983.

\bibitem{DBLP:conf/stoc/Karger20}
David~R. Karger.
\newblock A phase transition and a quadratic time unbiased estimator for
  network reliability.
\newblock In {\em Proccedings of the 52nd Annual {ACM} {SIGACT} Symposium on
  Theory of Computing, {STOC} 2020, Chicago, IL, USA, June 22-26, 2020}, pages
  485--495, 2020.

\bibitem{DBLP:journals/jcss/KhotR08}
Subhash Khot and Oded Regev.
\newblock {Vertex cover might be hard to approximate to within $2-\epsilon$}.
\newblock {\em J. Comput. Syst. Sci.}, 74(3):335--349, 2008.

\bibitem{khursheed1981positive}
Alam Khursheed and KM~Lai~Saxena.
\newblock Positive dependence in multivariate distributions.
\newblock {\em Communications in Statistics-Theory and Methods},
  10(12):1183--1196, 1981.

\bibitem{DBLP:journals/rsa/Vondrak07}
Jan Vondr{\'{a}}k.
\newblock Shortest-path metric approximation for random subgraphs.
\newblock {\em Random Struct. Algorithms}, 30(1-2):95--104, 2007.

\bibitem{YM18}
Yutaro Yamaguchi and Takanori Maehara.
\newblock {Stochastic Packing Integer Programs with Few Queries}.
\newblock In {\em Proceedings of the Twenty-Ninth Annual {ACM-SIAM} Symposium
  on Discrete Algorithms, {SODA} 2018, New Orleans, LA, USA, January 7-10,
  2018}, pages 293--310, 2018.

\end{thebibliography}
	
\appendix
\section{Deferred Proofs}\label{sec:proofs}

\subsection{Polynomial-Time Implementation of Algorithm~\ref{alg:bipartite}}\label{apx:polytime}

Algorithm~\ref{alg:bipartite} which proves Theorem~\ref{thm:bipartitevc} relies on a partitioning of Lemma~\ref{lem:vcpartition}. The described proof for this lemma, as stated above, is not through a polynomial-time construction. In this section, we address this issue and explain how the guarantee of Lemma~\ref{lem:vcpartition} can also be achieved in polynomial (randomized) time, which leads to Algorithm~\ref{alg:bipartite} running in polynomial-time.

The reason that our algorithm for finding the partitioning $(Q, S)$ of Lemma~\ref{lem:vcpartition} is not polynomial-time, is that we assume each matching algorithm $\SM{i}$ maximizes the quadratic objective (\ref{eq:objective}) which is not clear how to do in polynomial time. Here, however, we show how it is possible to get around this barrier and achieve the same guarantee in polynomial time too.

Suppose that we define partitionings $(Q_i, S_i)$ as before, but instead of using a matching algorithm $\SM{i}$ for each which maximizes $\Phi_i(\SM{i})$, we use simply an arbitrary maximum matching algorithm $\mc{M}_i'$. What can go wrong? We used the assumption that $\SM{i}$ maximizes $\Phi_i(\SM{i})$ in Claim~\ref{cl:matchingsdecreasing} which led to Claim~\ref{cl:interval}, proving existence of $\Omega(\epsilon^6 p^3 k)$ partitionings where every pair of them achieve the same objective up to an additive $O(\epsilon^6 p^3 \mu(G))$ difference.  It was also used crucially in the proof of Claim~\ref{cl:Dicontinuestocontribute}, where we constructed an algorithm $\SM{i, j}$ which returns either the output of $\SM{i}(H_i)$ or that of $\SM{j}(H_j)$ each with probability $1/2$. There we argued that if the desired bound of Claim~\ref{cl:Dicontinuestocontribute} does not hold, then $\Phi_i(\SM{i, j}) > \Phi_i(\SM{i}) + 0.01 \epsilon^6 p^3 \mu(G)$ (see (\ref{eq:htcr1823973})),  contradicting the assumption that $\Phi_i(\SM{i})$ is the maximum possible achievable objective. However, in order for this to lead to a contradiction, we do not necessarily need $\Phi_i(\SM{i})$ to have the maximum possible value. Rather, it is sufficient to merely guarantee $\Phi_i(\SM{i, j}) < \Phi_i(\SM{i}) + 0.01 \epsilon^6 p^3 \mu(G)$ for all $i < j$. Also in order to guarantee Claim~\ref{cl:interval}, it suffices to have, say, $\Phi_i(\SM{i}) \geq \Phi_j(\SM{j}) \pm O((\epsilon p)^{10})$ for all $i < j$.

To achieve this guarantee, we first present in Claim~\ref{cl:estimator} a polynomial-time randomized algorithm for estimating $\Phi_i(\mc{M}')$ with $O(1)$ additive error, for any given polynomial-time algorithm $\mc{M}'$.

\begin{claim}\label{cl:estimator}
	Given any matching algorithm $\mc{M}'$, it is possible to estimate the value of $\Phi_i(\mc{M}')$ (for any $i$) with $O(1)$ additive error in polynomial time, with high probability.
\end{claim}
\begin{proof}
	For each edge $e$, let $p_e := \Pr[e \in \mc{M}'(H_i)]$ and recall from (\ref{eq:objective}) that 
	$
	\Phi_i(\mc{M}') = \sum_{e \in E} p_e - \epsilon p_e^2.
	$
	We do not know the value of $p_e$ but can estimate it by random sampling.  Take $t = n^6$ independent outputs $M_1, \ldots, M_t$ of the random matching $\mc{M}'(H_i)$. Let $X_e := \sum_{i = 1}^t \pmb{1}(e \in M_i)$ be the number of these matchings that include $e$. Since $\E X_e = t \cdot p_e$, a simple application of Chernoff bound gives
	$$
		\Pr\left[|X_e - t p_e| \geq \sqrt{2 t \ln n}\right] \leq 2e^{-4 \ln n} \leq 2n^{-4}.
	$$
	Therefore, by letting $q_e := \frac{1}{t} X_e$ and a union bound over the less than $n^2/2$ choices of $e$, we get that with probability at least $1- n^{-2}$, for any edge $e$ it holds that $|q_e - p_e| < \frac{\sqrt{2t \ln n}}{t} < \frac{1}{n^2}$.
	
	Using this estimator $q_e$ instead of $p_e$, with probability $1-n^{-2}$ we get the following estimator with constant additive error
	$$
		\sum_{e \in E} q_e - \epsilon q_e^2 =  \Phi_i(\mc{M}') \pm O\left(n^2 \cdot \frac{1}{n^2}\right) = \Phi_i(\mc{M}') \pm O(1),
	$$
	completing the proof.
\end{proof}

Having this estimator, we then use an arbitrary maximum matching algorithm $\mc{M}_i'$ for each partitioning $(Q_i, S_i)$ and based on that construct the next partitioning $(Q_{i+1}, S_{i+1})$. Then for a margin $\delta = \Theta((\epsilon p)^{10} \mu(G))$, if it happens for some $j > i$, that our estimator predicts $\Phi_j(\mc{M'}_j) > \Phi_i(\mc{M'}_i) + \delta$ or $\Phi_i(\mc{M'}_{i, j}) > \Phi_i(\mc{M'}_i) + \delta$ (where $\Phi_i(\mc{M'}_{i, j})$ returns the matching $\Phi_i(\mc{M'}_{i})$ with probability $1/2$ and $\Phi_j(\mc{M'}_{j})$ otherwise), we simply use that algorithm instead of $\mc{M'}_i$ for $i$. Note that using this new algorithm may cause $D_i$ to change and so we may need to re-construct the the partitionings $(Q_{i+1}, S_{i+1}), \ldots, (Q_k, S_k)$. 

Finally, we argue why this process stops after polynomially many iterations. Every time that we change the matching algorithm of a partitioning $(Q_i, S_i)$, its objective $\Phi_i$ (and not just its estimation) increases by $O((\epsilon p)^{10}\mu(G))$. On the other hand, by definition (\ref{eq:objective}), the value of $\Phi_i$ for any $i$ is upper bounded by the maximum matching $\mu(G)$ of $G$. Hence, this $\Theta((\epsilon p)^{10}\mu(G))$ increase in $\Phi_i$ can only happen for at most $O((\epsilon p)^{-10})$ steps for any fixed partitioning $(Q_i, S_i)$. This, in turn, implies that the total number of changing the matching algorithm for any of the $k$ partitionings is bounded by $k^{O(1/\epsilon^{10}p^{10})} = O_{\epsilon, p}(1)$. As a result, the whole process takes $O_{\epsilon, p}(1) \cdot \poly(n)$ time.

\subsection{Proof of Claim~\ref{cl:yv-concentrate}}

In this section, we prove Claim~\ref{cl:yv-concentrate}. We start with the notation we use in the proof.

\smparagraph{Notation:} We fix an arbitrary vertex $v \in A$ and prove Claim~\ref{cl:yv-concentrate} for it. We use $u_1, \ldots, u_d$ to denote the neighbors of $v$ in $S$, and denote $e_j = (u_j, v)$. We also let $A' = \{v_1, \ldots, v_{n'}\} = A \setminus \{v\}$ be the set of all vertices in $A$ excluding $v$.

We start with an auxiliary claim that will be helpful both in bounding the expected value of random variable $(y_v \mid \notvprop)$ and also proving a concentration bound for it.

\begin{claim}\label{cl:lllcccgg-h1230}
	It holds that $\frac{1}{\Pr[\notvprop]} \sum_{i=1}^d q_{e_i} \leq q_v$.
\end{claim}
\begin{proof}
	Observe from  Corrolary~\ref{cor:dchuaoe18279312873} Property $(ii)$ that $\Pr[\vprop] = (1-\epsilon) \Pr[v \in \mc{M}(Q_p)]$. Thus,
	$$
	\Pr[\notvprop] = 1-\Pr[\vprop] = 1 - (1-\epsilon) \Pr[v \in \mc{M}(Q_p)] \geq 1 - \Pr[v \in \mc{M}(Q_p)].
	$$
	As a result, we get
	$$
		\frac{\sum_{i=1}^d q_{e_i}}{\Pr[\notvprop]} \leq \frac{\sum_{i=1}^d q_{e_i}}{1 - \Pr[v \in \mc{M}(Q_p)]}.
	$$
	Let $q^S_v$ denote the sum of $q_e$'s written on edges $e \in S$ connected to $v$ and let $q^Q_v$ denote the same but on edges of $v$ in $Q$. Observe that the nominator of the fraction above is exactly $q_v^S$ and the denominator is $1-q^Q_v$ by the first assumption of Lemma~\ref{lem:appvil}. Combined with $q^S_v + q^Q_v = q_v$ and $q_v \leq 1$ (since $\b{q}$ is a valid fractional matching by  Lemma~\ref{lem:appvil}) we get
	\begin{equation}
\frac{\sum_{i=1}^d q_{e_i}}{1 - \Pr[v \in \mc{M}(Q_p)]} =  \frac{q^S_v}{1-q^Q_v} = \frac{q_v \cdot q^S_v}{q_v (1-q_v^Q)} = \frac{q_v \cdot q^S_v}{q_v-q_v \cdot q_v^Q}  \leq \frac{q_v \cdot q^S_v}{q_v - q^Q_v} = \frac{q_v \cdot q^S_v}{q^S_v} = q_v,
\end{equation} 
	which is the stated bound.
\end{proof}

Let us first bound the expected value of $y_v$ conditioned on event $(\notvprop)$.

\begin{claim}\label{cl:expyv}
	Let $v$ be as above. Then $\E[y_v \mid \notvprop] \leq q_v$.
\end{claim}
\begin{proof}
	We have
	$$
		\E[y_v \mid \notvprop] = \sum_{i=1}^d \frac{q_{e_i}}{p \Pr[u_i \not\in M_\mc{B}] \Pr[\notvprop]} \cdot \Pr[e_i \in S_p, u \not\in M_\mc{B}, \notvprop \mid \notvprop].
	$$
	Event $e_i \in S_p$ is independent of $u \not\in M_\mc{B}, \notvprop$ as discussed before and $\notvprop$ and $u \not\in M_\mc{B}$ are also independent by Corollary~\ref{cor:dchuaoe18279312873} Property $(iv)$. This means 
	$$
		\E[y_v \mid \notvprop] = \sum_{i=1}^d \frac{q_{e_i}}{p \Pr[u_i \not\in M_\mc{B}] \Pr[\notvprop]} \cdot p \Pr[u_i \not\in M_\mc{B}] = \frac{1}{\Pr[\notvprop]} \sum_{i=1}^d q_{e_i}.
	$$
	Applying Claim~\ref{cl:lllcccgg-h1230} on the RHS concludes the claim.
\end{proof}

\begin{claim}\label{cl:covv}
	For any pair of edges $e_i=(v, u_i)$ and $e_j=(v, u_j)$ with $e_i \not= e_j$ and  $\{u_i,u_j\} \subseteq \{u_1, \dots, u_d\}$, let  $y'_{e_i} = (y_{e_i} \mid \notvprop)$ and $ y'_{e_j}=(y_{e_j} \mid \notvprop)$. We have $\Cov(y'_{e_i}, y'_{e_j})\leq 0$.
\end{claim}
\begin{proof}
To prove this claim, we will use some known facts about negatively associated (NA) random variables. By definition, a set of random variables are NA, if any two monotone nondecreasing functions $f$ and $g$ defined on disjoint subsets of them satisfy $\E[f.g]\leq \E[g].\E[f].$ Below are three facts about negative association based on \cite{khursheed1981positive, joag1983negative, dubhashi1996balls}.
\begin{enumerate}
\item[(1)] Any set of Bernoulli random variables whose sum is upper-bounded by one are NA. 
\item[(2)] If $A$ is a set of NA random variables, $B$ is a set of NA random variables with $A$ and $B$ independent of each other, $A \cup B$ is also a set of NA random variables.	
\item[(3)] Let $X=\{x_1, \dots, x_m\}$ be a set of NA random variables. If $f_1, \dots, f_k$ are a set of monotone nondecreasing functions defined on disjoint subsets of $X$, then $f_1, \dots, f_k$ are NA.
\item[(4)] Let $\{x_1, \dots x_m\}$ be a set of NA random variables. Then, for any $i \not= j$, $\Cov(x_i, x_j)\leq 0.$
\end{enumerate}

We will start by showing that random variables $y'_{e_1}, \dots, y'_{e_d}$ are NA where for any $i\in [d]$, we define $ y'_{e_i}=(y_{e_i} \mid \notvprop)$. For the rest of the proof, we will omit the condition \notvprop{} from all the statements for simplicity.

For any pair of vertices $i\in A/\{v\}$ and $j\in B$, let us define Bernoulli random variable $x_{i,j}$ to be equal to one iff vertex $i$ sends a proposal to vertex $j$. Note that for any $i$, we have $\sum_{j\in B} x_{i,j}\leq 1$. Also, since vertices send their proposals independently from each other,  invoking the first two facts above implies $\{x_{i,j} : \forall i\in A/ \{v\}, j\in B\}$ is a set of NA Bernoulli random variables. Now, for any vertex $u_j\in\{u_1, \dots, u_d\}$, we define random variable $z_j$ to be equal to one if $u_j \in M_\mc{B}$. Note that we have $z_j=1$ iff $j$ receives at least one proposal from $A/\{v\}$ and so it is a monotone nondecreasing function of  $\{x_{i,j} : i\in A/\{v\}\}$. Since for any $z_j$ and $z_j'$ these subsets are disjoint, invoking the third fact implies $z_i, \dots, z_d$ are NA. 
Finally, for any edge, we know $y'_{e_i}$ is a monotone nondecreasing function of $z_i$ and whether $e_i$ is realized. Based on the third fact, this implies negative association of  $y'_{e_1}, \dots, y'_{e_d}$. Thus, by the fourth property $\Cov(y'_{e_i}, y'_{e_j})\leq 0$ for any $i \not= j$.
\end{proof}

We are now ready to prove Claim~\ref{cl:yv-concentrate}.

\begin{proof}
	As proved in Claim~\ref{cl:u-frac-exp}, $\E[y_v \mid \notvprop] \leq q_v$. We prove the desired inequality of the claim via a concentration bound on random variable $y'_v := (y_v \mid \notvprop)$. Let us for simplicity also define random variable $y'_{e_i} := (y_e \mid \notvprop)$ and observe that
	$y'_v = \sum_{i=1}^d y'_{e_i}.$ 
	
	By Claim~\ref{cl:covv} we know $\Cov(y'_{e_i}, y'_{e_j})\leq 0$ for any $i \not= j$. We can thus bound the variance of $y'_v$ as follows:
	\begin{align*}
		\Var[y'_v] &= \sum_{i=1}^d \Var[y'_{e_i}] + 2\sum_{1\leq i<j\leq d} \Cov(y'_{e_i}, y'_{e_j}) \\
		&\leq \sum_{i=1}^d \E[(y'_{e_i})^2] - \E[y'_{e_i}]^2 \leq \sum_{i=1}^d \E[(y'_{e_i})^2]  \stackrel{\text{see below}}{\leq} \tau \cdot \E[y'_v]  \stackrel{\text{Claim~\ref{cl:expyv}}}{\leq} \tau \cdot q_v,
	\end{align*}
	where $\tau$ is the maximum possible outcome of $y'_{e_i}$ for any $i \in [d]$. 
	
	Plugging this into Chebyshev's inequality, we get
	$$
		\Pr[y'_v > \E[y'_v]+\delta] \leq \frac{\Var[y'_v]}{\delta^2} \leq \frac{\tau q_v}{\delta^2}.
	$$
	Finally, for each edge $e_i$, by construction of $\b{y}$, $y'_{e_i}\leq y_{e_i} \leq q_{e_i}/(p\Pr[\notxprop{v}]\Pr[u_i \not\in M_\mc{B}]) \leq q_{e_i}/p\epsilon^2$ where the latter follows from Corollary~\ref{cor:dchuaoe18279312873}. Combined with $q_{e_i} \leq \epsilon^5 p$ by Lemma~\ref{lem:appvil} and $e_i \in S$, we get $\tau \leq \epsilon^3$. We thus get $\Pr[y'_v > 1+\epsilon] \leq \Pr[y'_v > \E[y'_v]+\epsilon] \leq \frac{\epsilon^3}{\epsilon^2} q_v$. This concentration of $y'_v$ implies $\E[y_v \mid y_v \leq 1+\epsilon,\notvprop] \leq \epsilon q_v$ and thus the stated bound of the claim.
\end{proof}

\end{document}